\documentclass[12pt]{article}
\usepackage[utf8]{inputenc}
\usepackage[english]{babel}
\usepackage{a4wide,authblk,cite}
\usepackage{graphicx}
\usepackage{color}
\usepackage{forloop}
\usepackage{amsmath}
\usepackage[version=3]{mhchem}
\usepackage{dsfont}
\usepackage{amssymb}
\usepackage{amsthm}
\usepackage{enumerate}
\usepackage{txfonts}
\usepackage[fleqn,tbtags]{mathtools}
\usepackage{array}
\usepackage{framed}
\usepackage{fancybox}
\usepackage{comment}
\usepackage{ifthen}
\usepackage{fancyhdr}
\usepackage{todonotes}

\newcounter{thm}
\numberwithin{equation}{section}

\newtheorem{theorem}[thm]{Theorem}
\newtheorem{lemma}{Lemma}
\newtheorem{prop}{Proposition}

\theoremstyle{definition}
\newtheorem{remark}{Remark}
\newtheorem*{remark*}{Remark}
\newtheorem*{example*}{Example}

\newcommand{\linePage} {\noindent\makebox[\linewidth]{\rule{\textwidth}{1pt}} \\}
\specialcomment{simple} {$ $ \\ \begingroup   \scriptsize Simple Calculation \vspace{-4mm}   \\ \linePage   \endgroup \begingroup   }{$ $ \vspace{-3mm} \\ \linePage  \endgroup}

\excludecomment{simple}

\newcommand\con{{\mathcal{C}}}

\newcommand{\dI} { \,\mathrm{d}}

\newcommand{\norm}[1]{\left\| #1 \right\|}

\newcommand{\R}{\mathbb{R}}
\newcommand{\Ss}{\mathbb{S}}

\renewcommand{\Re}{\operatorname{Re}}

\newcommand{\rk} {\right}
\newcommand{\lk} {\left}

\newcommand{\const}{{ \rm const\,}}

\newcommand{\Tdiag}{T_{\rm{diag} }}
\newcommand{\Toff}{T_{\rm{off} }}

\newcommand\calF{\mathcal{F}}
\newcommand\calG{\mathcal{G}}
\newcommand\calT{\mathcal{T}}
\newcommand\calL{\mathcal{L}}
\newcommand\calH{\mathcal{H}}
\newcommand\calGam{\tilde{\Gamma}}

\begin{document}

\setlength{\parskip}{0.6em} 
\setlength{\parindent}{1.5em}

\title{Stability of a Fermionic $N+1$ Particle System with Point Interactions}
\author{Thomas Moser and Robert Seiringer}
\affil{IST Austria, Am Campus 1, 3400 Klosterneuburg, Austria}
\date{June 28, 2017}

\maketitle

\begin{abstract}
We prove that a system of $N$ fermions interacting  with an additional particle via point interactions is stable if the ratio of the mass of the additional particle to the one of the fermions is larger than some critical  $m^*$. The value of $m^*$ is independent of $N$ and turns out to be less than $1$. This fact has important implications for the stability of the unitary Fermi gas. We also characterize the domain of the  Hamiltonian of this model, and   establish the validity of the Tan relations for all wave functions in the domain.
\end{abstract}

\section{Introduction}

Models of particles with point interactions are ubiquitously used in physics, as an idealized description whenever the range of the interparticle interactions is much shorter than other relevant length scales. They were introduced in the early days of quantum mechanics as models of nuclear interactions \cite{wigner,bethe,Thomas1935,fermi}, but have proved useful  in other branches of physics, like polarons  (see \cite{Massignan} and references there) and cold atomic gases \cite{zwerger}. While the two-particle problem is mathematically completely understood \cite{albb}, for more than two particles  the existence of a self-adjoint  Hamiltonian that is bounded from below and models pairwise point interactions is a challenging open problem. It is known that such a Hamiltonian can only exist for fermions with at most two components (or two different species of fermions), due do the Thomas effect  \cite{Thomas1935,Braaten06,Tamura1991,Yafaev74}. 

For $N\geq 2$, we consider here a system of $N$ (spinless) fermions of mass $1$, interacting with another particle of mass $m$ via point interactions. The latter are characterized by a parameter $\alpha\in \R$, where $-1/\alpha$ is proportional to the scattering length of the pair interaction \cite{albb}. 
Purely formally, the Hamiltonian of the system can be thought of as 
\begin{equation}\label{ham1}
 H = - \frac 1 {2 m} \Delta_{x_0} - \frac 1 2 \sum_{i=1}^N \Delta_{x_i} + \gamma \sum_{i=1}^N \delta(x_0-x_i)
\end{equation}
where $x_i \in \R^3$, and $\gamma$ represents an infinitesimal coupling constant. Models of this kind have been studied extensively  in the literature (see, e.g., \cite{Castin10,correggi2012stability, Correggi2015, Correggi15, Dell, Dimock, Faddeev,Finco2012, Mich13, Mich16, MichX, Minlos2011, Minlos2012, Minlos2014,Ter,werner}) and can be defined via a suitable regularization procedure. More precisely, the formal expression \eqref{ham1} can be given a  meaning in terms of a suitable quadratic form \cite{correggi2012stability, Dell, Finco2012}, which will be introduced in the next section. However, only in case the quadratic form is stable, i.e.,  bounded from below,  does it give rise to a  unique self-adjoint operator and hence gives a precise meaning to \eqref{ham1}. We are interested in this question of stability. We shall show that there exists a critical mass $m^*$, independent of $N$, such that stability holds for $m>m^*$. The value of $m^*$
is determined by a two-dimensional optimization problem of a certain analytic function. 
A numerical evaluation of the  expression yields $m^*\approx 0.36$. 

In particular, the system under consideration is stable for $m=1$. This latter case is of particular importance, in view of constructing a model of a gas of spin $1/2$ fermions close to the unitary limit, where the scattering length becomes much larger than the range of the interactions. For $N+1$ such fermions, our result can be interpreted as proving the existence of such a model in the sector of total spin $(N-1)/2$, i.e., $1$ less than the maximal value. Of course stability holds trivially in the sector of total spin $(N+1)/2$, since the particles do not interact in this case due to the total antisymmetry of the spatial part of the wave functions. We note that stability in other spin sectors is still an open problem, whose solution would be of great interest because of the relevance of the model for cold atomic gases (see \cite{zwerger} and references there). For its solution, it is necessary to understand the problem of stability for general systems of $N+M$ particles mutually interacting via point interactions. In the case $N=M=2$, a numerical analysis  suggests stability, see \cite{Mich16} for the case $m=1$ and \cite{Endo} for the full range of mass ratios where stability for  the $2+1$ problem holds, i.e., for $0.0735 < m < (0.0735)^{-1} \approx 13.6$ \cite{Braaten06}.

\section{Model and Main Results}

Because of translation invariance,  it is convenient to separate the center-of-mass motion and to introduce relative coordinates $X =  \lk(m x_0+\sum_{i=1}^N x_i \rk)/(m+N)$, $y_i = x_i-x_0$ 
for $1\leq i \leq N$ in the usual way. With their aid we can formally write the operator $H$ in \eqref{ham1} as $ H = H_{\rm{cm}} + \frac {m+1}{2m} H_{\rm rel} $, 
where $H_{\rm cm} = -(2(m+N))^{-1} \Delta_X$ and
\begin{equation}\label{ham}
H_{\rm rel} = - \sum_{i=1}^N \Delta_{y_i} - \frac 2{m+1} \sum_{1\leq i < j\leq N} \nabla_{y_i} \cdot \nabla_{y_j} + \tilde\gamma \sum_{i=1}^N \delta(y_i) 
\end{equation}
for $\tilde\gamma = 2m\gamma/(m+1)$. The latter operator acts on purely anti-symmetric functions of $N$ variables only.

The formal expression \eqref{ham} can be given a  meaning in terms of a suitable quadratic form \cite{correggi2012stability, Dell, Finco2012}, which will be introduced in the next subsection.

\subsection{Quadratic Form and Stability}
The model under consideration here is defined via a quadratic form $F_\alpha$ as follows. 
For $\mu>0$ and $q_i\in \R^3$, $1\leq i\leq N$, let 
\begin{equation}\label{def:G}
G(q_1, \ldots, q_N) \coloneqq \left ( \sum_{i=1}^N q_i^2 +\frac 2 {m+1}  \sum_{1\leq i<j\leq N} q_i \cdot q_j + \mu \right)^{-1} 
\end{equation}
The quadratic form $F_\alpha$ has the domain
\begin{equation}\label{def:DF}
D(F_\alpha) = \left\{ u \in L_{\rm as}^2(\R^{3 N}) \mid u = w + G \xi , w \in H_{\rm as}^1(\R^{3 N}), \xi \in H_{\rm as}^{1/2}(\R^{3 (N-1)})\right\} 
\end{equation}
where $G\xi$ is short for the function with Fourier transform  
\begin{equation}
\widehat{G \xi} (q_1,\ldots, q_N) =  G(q_1, \ldots, q_N) \sum_{i=1}^N (-1)^{i+1}  \hat \xi(q_1, \ldots ,q_{i-1}, q_{i+1}, \ldots, q_N)
\label{eq:quadraticForm} 
\end{equation}
and the subscript ``as'' indicates functions that are antisymmetric under permutations. For $u\in D(F_\alpha)$, we have
\begin{align}\nonumber
F_\alpha(u) & = \left\langle w \left|  - \sum\nolimits_{i=1}^N \Delta_i  -\frac 2{m+1} \sum\nolimits_{1\leq i<j\leq N} \nabla_i \cdot \nabla_j  + \mu  \right| w\right\rangle  - \mu \norm{u}^2_{L^2(\R^{3N})} \\ & \quad + N \left( \alpha \norm{\xi}_{L^2(\R^{3(N-1)})}^2 +   \Tdiag(\xi)+  \Toff(\xi) \right)  \label{def:Fal}
\end{align}
where
\begin{align}\nonumber
\Tdiag(\xi) & \coloneqq \int_{\R^{3(N-1)}}  |\hat \xi(s,\vec q)|^2  L(s, \vec q) \dI s \dI \vec q  \\
\Toff(\xi) & \coloneqq (N-1) \int_{\R^{3N}} \hat \xi^\ast (s ,\vec q) \hat \xi (t ,\vec q) G(s,t,\vec q)  \dI s \dI t \dI \vec q \label{def:T}
\end{align}
We introduced $\vec q \coloneqq (q_1,\dots,q_{N-2})$ for short,  and the function $L$ is given by 
\begin{equation}\label{def:L}
L(q_1, \ldots, q_{N-1}) \coloneqq 2 \pi^2 \lk ( \frac {m(m+2)}{(m+1)^2} \sum_{i=1}^{N-1} q_i^2 + \frac {2 m}{(m+1)^2} \sum_{1\leq i < j\leq N-1} q_i \cdot q_j + \mu \rk)^{1/2}
\end{equation}
Note that since $G\xi \not\in H^1(\R^{3N})$ for $\xi\neq 0$, the decomposition of $u$ as $u=w+G\xi$ is unique. Moreover, while $w$ depends on $\mu$,  $\xi$ is independent of the choice of $\mu$. 

Clearly $\Tdiag(\xi)$ is bounded above and below by $\|\xi\|_{H^{1/2}(\R^{3(N-1)})}^2$, and also $\Toff(\xi)$ is bounded in $H^{1/2}(\R^{3(N-1)})$ (see Sect.~\ref{sec:prelim}). 
One readily checks that both $D(F_\alpha)$ and $F_\alpha(u)$ are actually independent of $\mu$ for $\mu>0$, even though $\Tdiag(\xi)$ and $\Toff(\xi)$ depend on $\mu$. 
The domain $D(F_\alpha)$ is also independent of $\alpha\in \R$. Moreover, under the scaling $u \to u_\lambda(\,\cdot\,) = \lambda^{3N/2} u ( \lambda \, \cdot\,)$ for $\lambda>0$, $F_\alpha$ changes as $F_\alpha(u_\lambda) = \lambda^{2} F_{\lambda^{-1} \alpha}(u)$. In particular, $F_0$ is homogeneous of order $2$ under scaling.

The quadratic form $F_\alpha$ can be obtained as a limit of a suitably regularized version of \eqref{ham}, see \cite{Dell} and  \cite[Appendix~A]{correggi2012stability}. As we shall see in the next subsection, the parameter  $\alpha$ equals $-2\pi^2/a$, where $a$ denotes  the scattering length of the pair interaction. We note that other choices for quadratic forms are possible in the unitary case $\alpha=0$ for small mass $m$, see \cite{Correggi2015}.

To state our main result, we define, for any $m>0$, 
\begin{align}\nonumber &
\Lambda(m)= \sup_{s,K\in\R^3, Q>0}  \frac{s^2 + Q^2}{\pi^2 (1+m)} \ell_m(s,K,Q)^{-1/2}  \int_{\R^3}  \frac 1{t^2}   \ell_m(t,K,Q)^{-1/2}   \\ & \qquad\qquad \times  \frac { \left| (s+AK)\cdot (t+AK) \right| }{ \left[ (s+AK)^2 + (t+AK)^2 + \frac m{1+m}( Q^2 +A K^2) \right]^2  - \left[ \frac 2{(1+m)} (s+AK)\cdot(t+AK)\right]^2}    \dI t \label{def:Lm}
\end{align}
where $A:=(2+m)^{-1}$ and 
\begin{equation}\label{def:lm}
\ell_m(s,K,Q) := \left( \frac{m}{(m+1)^2} (s+K)^2 + \frac m{m+1} (s^2 + Q^2)  \right)^{1/2}
\end{equation}
A somewhat simpler, equivalent expression for $\Lambda(m)$, involving only the supremum over two positive parameters, will be given in Section~\ref{sec:numerics}. 
We shall show in Section~\ref{sec:Lambda}  that $\Lambda(m)$ is finite, and satisfies the upper bound
\begin{equation}\label{bound:L}
\Lambda(m) \leq   \frac{ 4 (1+m)^2 (2+4m +m^2)^{3/2}}{\sqrt{2} \pi  \left[ m(m+2) \right]^{3} }  
\end{equation}
Note  that \eqref{bound:L} implies, in particular, that $\lim_{m\to\infty} \Lambda(m) = 0$.

Our first main result is the following:

\begin{theorem}\label{thm1}
For any $   \xi \in H_{\rm as}^{1/2}(\R^{3 (N-1)})$, $\mu>0$ and  $N\geq 2$, 
\begin{equation}\label{eq:thm}
 \Toff(\xi) \geq  -\Lambda(m)   \Tdiag(\xi)
\end{equation}
In particular, if $m$ is such that $\Lambda(m)<1$, then $F_\alpha$ is closed and bounded from below by 
\begin{equation}\label{thm:lb}
F_\alpha(u) \geq \left\{ \begin{array}{ll} 0 & \text{for $\alpha\geq 0$} \\  - \left( \frac {\alpha}{2 \pi^2(1- \Lambda(m))} \right)^2 \|u\|_{L^2(\R^{3N})}^2 & \text{for $\alpha<0$} \end{array}\right.
\end{equation}
for all $u\in D(F_\alpha)$.
\end{theorem}

We note that \eqref{thm:lb} follows immediately from \eqref{eq:thm} in combination with the simple estimate $\Tdiag(\xi)\geq 2\pi^2\sqrt{\mu}\|\xi\|^2_{L^2(\R^{3(N-1)})}$. For $\alpha<0$, one simply chooses $\mu = \alpha^2 (2\pi^2(1-\Lambda(m))^{-2}$, using the independence of $F_\alpha(u)$ of $\mu$. As a closed and bounded from below quadratic form, $F_\alpha$ 
gives rise to a unique self-adjoint operator \cite[Thm.~VIII.15]{RS1} for $\Lambda(m)<1$. 
We shall describe it in detail in the next subsection.

The lower bound \eqref{thm:lb} is sharp as $m\to \infty$. For $\alpha<0$, $-(\alpha/2\pi^2)^2$ equals the binding energy of the two-particle problem with point interactions. As $m\to \infty$, only one of the fermions can be bound, hence the ground state energy becomes independent of $N$ in that limit. 

We emphasize that in contrast to the previous work \cite{correggi2012stability,Correggi15} we prove a bound on the critical mass that is independent of $N$ and, in particular, does not grow as $N$ gets large. Also the lower bound \eqref{thm:lb} is independent of $N$.

We shall prove Theorem~\ref{thm1} in Section~\ref{sec:proof} below. The right side of \eqref{bound:L}  turns out to be less than $1$ for $m\geq  1.76$, and hence stability holds in that region. For $m=1$, it equals about $2.47$, however, and is larger than $1$ as a result of the rather crude bounds leading to \eqref{bound:L}. 

In Section~\ref{sec:numerics} we  evaluate $\Lambda(m)$  numerically and show that it satisfies $\Lambda(1)<1$. In fact, from the numerics we shall see that $\Lambda(m)<1$ if $m\geq 0.36$ (see Fig.~\ref{fig:lambda1}). Recall that $F_\alpha$ is known to be unbounded from below \cite[Thm.~2.2]{correggi2012stability} for any $N\geq 2$ 
for $m\leq 0.0735$. In particular, the critical mass for stability satisfies $0.0735 < m^* <0.36$.

\begin{figure}
\begin{center}
\includegraphics[width = 11cm]{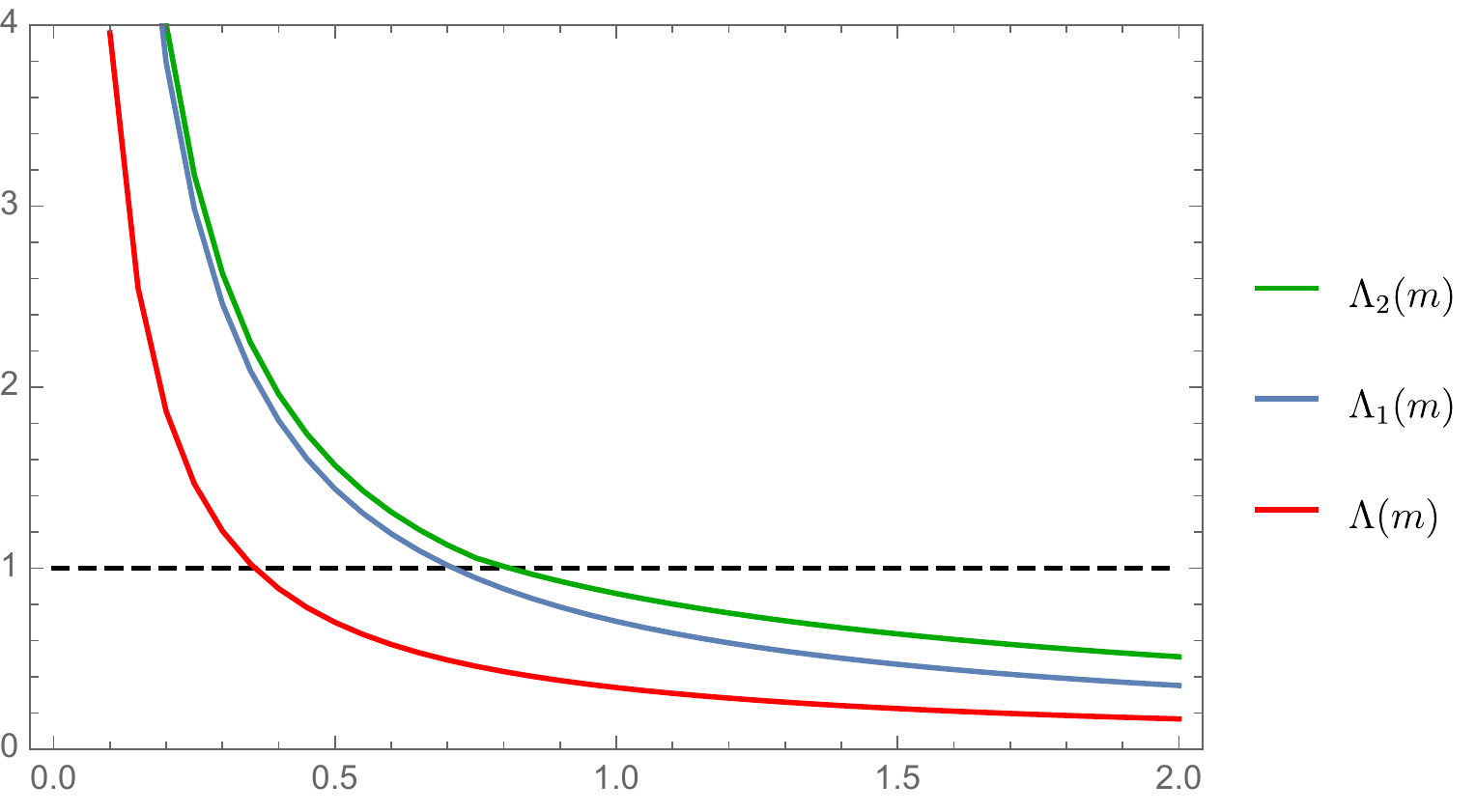}
\end{center}
\caption{Numerical evaluation of $\Lambda(m)$ defined in \eqref{def:Lm}. In the region $\Lambda(m)<1$, we prove stability of the system. Asymptotically, $\Lambda(m) \approx 1/(2 \sqrt{2} m)$ for large $m$ (and in fact, approximately within a few percent in the whole region $m\gtrsim 1$). For $\Lambda_1(m)<1$, we prove that the domain of the operator $\Gamma$ in \eqref{def:Gamma} equals $H^1_{\rm as}(\R^{3(N-1)})$. Moreover, for $\Lambda_2(m)<1$ the boundary condition in \eqref{def:DHa} implies that for every function in the domain of $H_\alpha$ one has $\xi \in H^{3/2}_{\rm as}(\R^{3(N-1)})$.
 }\label{fig:lambda1}
\end{figure}

\subsection{Hamiltonian}

For $\Lambda(m)<1$, Theorem \ref{thm1} implies  that 
\begin{equation}\label{def:Gamma}
 \Tdiag(\xi)+  \Toff(\xi)  = \langle \xi | \Gamma \xi \rangle
\end{equation}
defines a positive selfadjoint operator $\Gamma$ on  $L^2_{\rm as}(\R^{3 (N-1)})$, with domain $D(\Gamma) \subset H_{\rm as}^{1/2}(\R^{3 (N-1)})$. In fact,
\begin{equation}
\Gamma \geq \left(1-\Lambda(m)\right) L \geq \left(1-\Lambda(m)\right) 2\pi^2 \sqrt{\mu}
\end{equation}
where $L$ is short for the multiplication operator in momentum space defined by \eqref{def:L}.  


It is  not difficult to see that $H^{1}_{\rm as}(\R^{3 (N-1)})\subset D(\Gamma)$ (see Sect.~\ref{sec:prelim}), but this inclusion could possibly be strict. In fact, it was shown in \cite{Minlos2012,Minlos2014b} in the case $N=2$ that  $\Gamma$ is not selfadjoint on $H^1$ for certain small $m$, but admits a one-parameter family of semi-bounded self-adjoint extensions. In contrast, 
the following theorem implies that $D(\Gamma) = H^{1}_{\rm as}(\R^{3 (N-1)})$ for larger $m$, more precisely for $\Lambda_1(m)<1$, which is slightly more restrictive than our regime of stability, $\Lambda(m)<1$. 

To state our result, we define, analogously to \eqref{def:Lm}, for $\beta\geq 0$ and $m>0$, 
\begin{align}\nonumber &
\Lambda_\beta (m)= \sup_{s,K\in\R^3, Q>0}  \frac{s^2 + Q^2}{\pi^2 (1+m)}   \int_{\R^3}  \frac 1 {t^2} \left( \frac  {\ell_m(s,K,Q)^{(\beta-1)/2 } }{ \ell_m(t,K,Q)^{(\beta +1)/2 } }+ \frac  {\ell_m(t,K,Q)^{(\beta-1)/2 } }{ \ell_m(s,K,Q)^{(\beta +1)/2 } }   \right) \\ & \qquad\qquad \times  \frac { \left| (s+AK)\cdot (t+AK) \right| }{ \left[ (s+AK)^2 + (t+AK)^2 + \frac m{1+m}( Q^2 +A K^2) \right]^2  - \left[ \frac 2{(1+m)} (s+AK)\cdot(t+AK)\right]^2}    \dI t \label{def:Lmb}
\end{align}
Note that the integrand in \eqref{def:Lmb} is increasing and convex in $\beta$, hence 
$\Lambda_\beta(m)$ is, as a supremum over such functions, also increasing and convex. We have $\Lambda_\beta(m)\geq \Lambda_0(m) = 2  \Lambda(m)$. We shall show in Section~\ref{sec:Lambda} that $\Lambda_\beta(m)$ is finite for $\beta < 3$ and satisfies $\lim_{m\to\infty} \Lambda_\beta(m) = 0$. In particular, from the convexity it then follows that $\Lambda_\beta(m)$ is continuous in $\beta$ for $0\leq \beta < 3$. 

\begin{theorem}\label{thm:domain}
For any $\xi \in H^{1}_{\rm as}(\R^{3 (N-1)})$, $\mu>0$ and $N\geq 2$,
\begin{equation}
\| \Gamma \xi\|_{L^{2}(\R^{3 (N-1)})}^2 \geq \left(1 - \Lambda_1(m) \right) \| L \xi \|^2_{L^{2}(\R^{3 (N-1)})}
\end{equation}
In particular, if $\Lambda_1(m)<1$, then $D(\Gamma) = D(L) = H^{1}_{\rm as}(\R^{3 (N-1)})$. More generally, for $0\leq \beta \leq 2$, 
\begin{equation}\label{217}
\| L^{(\beta -1)/2} \Gamma \xi\|_{L^{2}(\R^{3 (N-1)})}^2 \geq \left(1 - \Lambda_\beta(m) \right) \| L^{(\beta+1)/2} \xi \|^2_{L^{2}(\R^{3 (N-1)})}
\end{equation}
for all $\xi \in H^{(\beta+1)/2}_{\rm as}(\R^{3 (N-1)})$.
\end{theorem}

The proof of Theorem~\ref{thm:domain} will be given in Section~\ref{sec:domain}.  A numerical evaluation of $\Lambda_\beta(m)$ yields $\Lambda_1(m)< 1$ for $m\geq 0.72$, while $\Lambda_{2}(m)<1$ for $m\geq 0.82$ (see Fig.~\ref{fig:lambda1}).

In terms of $D(\Gamma)$,  the self-adjoint operator $H_\alpha$ defined by the quadratic form $F_\alpha$ in \eqref{def:Fal} can be constructed in a straightforward way following the analogous construction in the two-dimensional case in \cite[Sect.~5]{Dell} (see also \cite{Teta,Finco2012,correggi2012stability,Minlos2012,Minlos2014b}). The result is 
 \begin{equation}\label{def:DHa}
D(H_\alpha) = \left\{ u \in L_{\rm as}^2(\R^{3 N}) \mid u = w + G \xi , w \in H_{\rm as}^2(\R^{3 N}), \xi \in D(\Gamma) , w\!\restriction_{y_N=0} = (2\pi)^{-3/2} (-1)^{N+1} (\alpha + \Gamma) \xi \right\} 
\end{equation}
and $H_\alpha$ acts on $u\in D(H_\alpha)$ as 
\begin{equation}
\left( H_\alpha + \mu \right) u =  \left(   - \sum_{i=1}^N \Delta_{y_i} - \frac 2{m+1} \sum_{1\leq i < j\leq N} \nabla_{y_i} \cdot \nabla_{y_j}   +\mu\right)  w 
\end{equation}

Note that as an $H^2$-function, $w$ has an $L^2$-restriction to the hyperplane $y_N=0$, and the last identity in \eqref{def:DHa} has to be understood as an identity of functions in $L^2_{\rm as}(\R^{3 (N-1)})$. In fact, the restriction of the $H^2$-function $w$ to the hyperplane $y_N=0$ is an $H^{1/2}$ function, and hence we conclude that for any $u \in D(H_\alpha)$, the corresponding $\xi$ satisfies $\Gamma\xi \in H^{1/2}$. The last part of Theorem~\ref{thm:domain} thus implies that for $\Lambda_{2}(m)<1$, $\xi$ is necessarily in $H^{3/2}$.

The last identity in \eqref{def:DHa} encodes the boundary condition satisfied by functions $u\in D(H_\alpha)$ at the origin. 
To see this, consider the behavior of the function $G\xi$ as $y_N\to 0$ or, equivalently, the integral of \eqref{eq:quadraticForm} over $q_N$ in a large ball. A short calculation using \eqref{eq:quadraticForm} shows that 
\begin{align}\nonumber
& \lim_{K\to\infty}  \int_{|q_N|<K} \left( \widehat{G \xi} (q_1,\ldots, q_N) - \frac 1{q_N^2} (-1)^{N+1}  \hat \xi(q_1, \ldots q_{N-1})\right) \dI q_N  
 \\ \nonumber & = \int_{\R^3} \left( G(q_1, \ldots, q_N) \sum_{i=1}^{N-1} (-1)^{i+1}  \hat \xi(q_1, \ldots ,q_{i-1}, q_{i+1}, \ldots, q_N) \right) \dI q_N  
 \\ \nonumber & \quad + (-1)^{N+1}  \hat \xi(q_1, \ldots q_{N-1})  \lim_{K\to\infty}  \int_{|q_N|<K}  \left( G (q_1,\ldots, q_N) - \frac 1{q_N^2} \right) \dI q_N \\ & = (-1)^N \widehat{  \Gamma\xi}(q_1,\ldots,q_{N-1})
\end{align}
where  we have used that
\begin{equation}
L(q_1, \dots , q_{N-1}) = - \lim_{K\to\infty} \int_{|q_N|<K} \left( G (q_1,\ldots, q_N) - \frac 1{q_N^2} \right) \dI q_N
\end{equation}
We conclude that the boundary condition in \eqref{def:DHa} implies that any $u\in D(H_\alpha)$ has the asymptotic behavior
\begin{equation}
\int_{|q_N|<K} \hat u(q_1,\ldots,q_N) \dI q_N  \approx \left(  4 \pi K   + \alpha\right)  {(-1)^{N+1} } \hat\xi(q_1,\ldots,q_{N-1}) \quad \text{as $ K \to \infty$.}
\end{equation}
In particular, $u$ diverges as $2\pi^2/|y_N|  + \alpha$ as $|y_N|\to 0$, and hence $\alpha$ is to be interpreted as  $\alpha = -2\pi^2 /a$ with $a$ the scattering length of the point interaction. A precise formulation of this divergence in configuration space will be given in Proposition~\ref{prop1} in the next subsection.

As in the case of the corresponding quadratic form, $H_\alpha$ is independent of the parameter $\mu$ used in its construction. Under a unitary scaling of the form $U_\lambda \psi(\, \cdot\,) = \lambda^{3(N+1)/2} \psi(\lambda\,\cdot\,)$, it transforms as $U_\lambda^{-1} H_\alpha U_\lambda = \lambda^2 H_{\lambda^{-1} \alpha}$.  Note that in contrast to $D(F_\alpha)$, the domain $D(H_\alpha)$ {\em does} depend on $\alpha$.

\subsection{Tan Relations}

In \cite{Tan}, Tan derived a number of identities that should hold for any  system of particles with point interactions (see also the review \cite{TanBraaten} and the references there). These can be experimentally tested, see \cite{hulet,vale1,vale2,jin,salomon}.  In this section, we shall present a rigorous version of the Tan relations for the Hamiltonian $H_\alpha$ constructed in the last subsection. The analysis in this section does not actually use the self-adjointness and analogous results  also hold for the general $N+M$ system, irrespective of its stability and the self-adjointness of the corresponding $H_\alpha$. We shall work with the assumption $\xi \in H^1$, however, which is guaranteed to be the case for $\Lambda_1(m)<1$, by Theorem~\ref{thm:domain}. 

In order to state the results, we have to re-introduce the center-of-mass motion. The Hilbert space for the $N+1$ system is thus  $L^2(\R^3) \otimes L^2_{\rm as}(\R^{3N})$, and the form domain of the corresponding quadratic form, which we denote by $\calF_\alpha$, equals
\begin{equation}
D(\calF_\alpha) = \left\{  \psi = \phi + \calG \xi \mid \phi \in H^1(\R^3)\otimes H_{\rm as}^1(\R^{3 N}), \xi \in H^{1/2}(\R^3)\otimes H_{\rm as}^{1/2}(\R^{3 (N-1)})\right\} 
\end{equation}
where
\begin{equation}\label{def:calG}
\calG(k_0,k_1, \ldots, k_N) \coloneqq \left ( \frac 1{2m} k_0^2 + \frac 12 \sum_{i=1}^N  k_i^2  + \mu \right)^{-1} \,,
\end{equation}
 $\calG\xi$ is short for the function with Fourier transform  
\begin{equation}
\widehat{\calG \xi} (k_0,k_1,\ldots, k_N) =  \calG(k_0, k_1, \ldots, k_N) \sum_{i=1}^N (-1)^{i+1}  \hat \xi(k_0 + k_i , k_1, \ldots ,k_{i-1}, k_{i+1}, \ldots, k_N)
\label{eq:tquadraticForm} 
\end{equation}
and, compared to \eqref{def:DF},  we have absorbed a factor $\frac{m+1}{2m}$ into the definition of $\xi$ for simplicity. 
For $\psi\in D(\calF_\alpha)$, we have
\begin{align}\nonumber
\calF_\alpha(\psi) & = \left\langle \phi \left| - \frac 1 {2 m} \Delta_{x_0} - \frac 1 2 \sum_{i=1}^N \Delta_{x_i}   + \mu \right| \phi\right\rangle  - \mu \norm{\psi}^2_{L^2(\R^{3(N+1)})} \\ & \quad + N  \left( \frac {2m }{m+1} \alpha \norm{\xi}_{L^2(\R^{3N})}^2 +   \calT_{\rm diag}(\xi)+  \calT_{\rm off}(\xi) \right)  \label{def:cFal}
\end{align}
where
\begin{align}\nonumber
\calT_{\rm diag}(\xi) & \coloneqq \int_{\R^{3N}}  |\hat \xi(k_0, k_1,\vec k)|^2  \calL(k_0, k_1, \vec k) \dI k_0 \dI k_1 \dI \vec k  \\
\calT_{\rm off}(\xi) & \coloneqq (N-1) \int_{\R^{3(N+1)}} \hat \xi^\ast (k_0+ s , t, \vec k) \hat \xi (k_0 + t , s, \vec k) \calG(k_0,s,t,\vec k)  \dI k_0 \dI s \dI t \dI \vec k \label{def:calT}
\end{align}
and we used    $\vec k = (k_2,\dots,k_{N-1})$ for short. The function $\calL$ is given by 
\begin{equation}\label{def:calL}
\calL(k_0, k_1, \ldots, k_{N-1}) \coloneqq 2 \pi^2 \left(\frac{2m}{m+1}\right)^{3/2}  \lk ( \frac {k_0^2}{2(m+1)} + \frac 12 \sum_{i=1}^{N-1} k_i^2 + \mu \rk)^{1/2}
\end{equation}

Theorem~\ref{thm1} implies that 
\begin{equation}\label{eq:bcal}
\calT_{\rm off}(\xi) \geq -\Lambda(m) \calT_{\rm diag}(\xi) \qquad \text{for all $\xi\in H^{1/2}(\R^3)\otimes H_{\rm as}^{1/2}(\R^{3 (N-1)})$.}
\end{equation}
To see this, one can either mimic the proof of Theorem~\ref{thm1}, or one simply argues as follows. Displaying the dependence on $\mu$ explicitly via a superscript in the expressions for $T_{\rm diag/off}$ and $\mathcal{T}_{\rm diag/off}$ in \eqref{def:T} and \eqref{def:calT}, respectively, it is straightforward to check that 
\begin{equation}
\mathcal{T}^\mu_{\rm diag/off}(\xi) = \frac {2m}{m+1} \int_{\R^3} T^{\tilde\mu_P} _{\rm diag/off}(\eta_P)\,  \dI P 
\end{equation}
where $\tilde \mu_P = \frac {2m}{m+1} ( \mu + \frac{P^2}{2(m+N)} )$ and 
\begin{equation}
\hat\eta_P(q_1,\dots,q_{N-1}) = \hat\xi \left( \tfrac{m+1}{m+N} P - \sum\nolimits_{j=1}^{N-1} q_j , q_1 + \tfrac 1{m+N} P , \dots, q_{N-1} + \tfrac 1{m+N} P \right)
\end{equation}
which is in $H^{1/2}_{\rm as}(\R^{3(N-1)})$ for almost every $P\in \R^3$. Since the bound \eqref{eq:thm} is uniform in $\mu$, \eqref{eq:bcal} follows.

Analogously to the discussion in the previous subsection, for $\Lambda(m)<1$ the quadratic form  $\calT_{\rm diag}(\xi)+  \calT_{\rm off}(\xi)$ defines a positive self-adjoint operator
$\calGam$ on $L^2(\R^3)\otimes L^2_{\rm as}(\R^{3 (N-1)})$. Explicitly, $\tilde\Gamma$ acts as
\begin{align}\nonumber
&\widehat{\tilde\Gamma \xi}(k_0,k_1,\ldots,k_{N-1}) \\ \nonumber & = \calL(k_0, k_1, \ldots, k_{N-1}) \hat \xi(k_0,k_1,\ldots,k_{N-1}) \\ &  \quad + \sum_{j=1}^{N-1} (-1)^{j+1} \int_{\R^3}  \calG(k_0-s,s,k_1,\ldots, k_{N-1}) \hat \xi(k_0  + k_j -s,s, k_1,\ldots,k_{j-1},k_{j+1},\ldots,k_{N-1}) \dI s
\end{align}
Theorem~\ref{thm:domain} implies that the domain $D(\tilde\Gamma)$ equals $H^{1}(\R^3)\otimes H_{\rm as}^{1}(\R^{3 (N-1)})$ in the case $\Lambda_1(m)<1$. 
 The domain of the self-adjoint operator $\calH_\alpha$ corresponding to the quadratic form $\calF_\alpha$ is given by those $\psi\in D(\calF_\alpha)$ where $\phi \in H^2(\R^3)\otimes H_{\rm as}^2(\R^{3 N})$, $\xi \in D(\calGam)$ and the boundary condition 
\begin{equation}\label{bc}
\phi\restriction_{x_N=x_0} = \frac{(-1)^{N+1}}{ (2\pi)^{3/2}} \left( \frac{2m\alpha}{m+1}+\calGam\right)\xi
\end{equation}
 is satisfied. The Hamiltonian $\calH_\alpha$ acts as
 \begin{equation}
\left( \calH_\alpha +\mu \right) \psi =  \left( - \frac 1 {2 m} \Delta_{x_0} - \frac 1 2 \sum_{i=1}^N \Delta_{x_i}   + \mu \right) \phi
\end{equation}
 It commutes with translations and rotations, and transforms  under scaling in  the same way as discussed for $H_\alpha$ at the end of the previous subsection.

The connection between the boundary condition \eqref{bc} and the asymptotic behavior of $\psi\in D(\calH_\alpha)$ as $|x_N-x_0|\to 0$ is explored in the following proposition, whose proof will be given in Section~\ref{sec:propp}.

\begin{prop}\label{prop1}
For any $\psi\in D(\calH_\alpha)$ with $\xi \in H^1(\R^{3N})$, we have
\begin{align}\nonumber 
\psi\left( R+ \tfrac r{1+m} ,x_1,\ldots, x_{N-1}, R - \tfrac{ mr}{1+m} \right) & = \left( \frac {2\pi^2} {|r|} + \alpha\right) \frac{2m}{m+1}\frac{(-1)^{N+1}}{(2\pi)^{3/2}} \xi(R,x_1,\ldots x_{N-1}) \\ & \quad + \upsilon(R,x_1,\ldots,x_{N-1},r) \label{eq:prop}
\end{align}
with $\upsilon(\,\cdot\, , r) \in L^{2}(\R^{3N})$ for all $r\in\R^3$, and $\lim_{r\to 0} \|\upsilon(\,\cdot\, , r)\|_{L^{2}(\R^{3N})} = 0$. 
\end{prop}

Proposition~\ref{prop1}  immediately implies a two-term asymptotics for the two-particle density
\begin{equation}
\rho(r) = N \int_{\R^{3N}} \left| \psi\left( R+ \tfrac r{1+m} ,x_1,\ldots, x_{N-1}, R - \tfrac{ mr}{1+m} \right) \right|^2 \dI R \dI x_1 \cdots \dI x_{N-1} 
\end{equation}
as $r\to 0$. In fact, $\rho$ satisfies  
\begin{equation}\label{def:rg}
\rho( r)=  \frac \pi 2 \left( \frac 1{|r|^2} - \frac 2{|r|a}  \right)  \con  + g(r) \quad \text{with}\quad \lim_{r \to 0} |r g(r) | = 0
\end{equation}
where $a = -2\pi^2/\alpha$ denotes the scattering length and 
\begin{equation}\label{def:con}
\con = \left( \frac{2m}{m+1}\right)^2 N \|\xi\|^2_{L^2(\R^{3N})}
\end{equation}
In the physics literature, $\con$ is called the {\em contact} \cite{Tan}. It turns out to play a crucial role in various other relevant quantities, as we shall demonstrate now.

For general $\psi \in L^2(\R^3) \otimes L^2_{\rm as}(\R^{3N})$,  the momentum densities of the mass $m$ (spin up) particle $n_\uparrow(k)$ and  of the mass $1$ (spin-down) particles $n_{\downarrow}(k)$ are defined as 
\begin{equation}
n_\uparrow(k) = \int_{\R^{3N}} | \hat\psi(k , k_1,\dots, k_N)|^2 \dI k_1 \cdots \dI k_N \ , \ n_\downarrow(k) = N \int_{\R^{3N}} |\hat\psi(k_0,k,k_2,\dots,k_N)|^2 \dI k_0 \dI k_2 \cdots \dI k_N 
\end{equation}
Our rigorous formulation of the Tan relation for the energy is as follows.

\begin{theorem}\label{thm:tan}
For $\psi \in D(\calH_\alpha)$ with $\xi \in H^1(\R^{3N})$, let $\con$ be given in \eqref{def:con}, 
and 
let 
\begin{equation}\label{defbark}
p_\uparrow = \frac{2m}{m+1} \| \xi\|_{L^2(\R^{3N})}^{-2} \int_{\R^{3N}}  k_1 |\hat\xi(k_1,\ldots,k_N)|^2 \dI k_1 \cdots \dI k_N \quad , \quad p_\downarrow= \frac 1m p_\uparrow.
\end{equation}
Then 
\begin{equation}\label{l1n}
k\mapsto k^2  n_\uparrow(k) - \frac \con  {|k-p_\uparrow|^{2}}  \in L^1(\R^3) \quad \text{and} \quad  k\mapsto k^2  n_\downarrow(k) - \frac \con  {|k-p_\downarrow|^{2}} \in L^1(\R^3)
\end{equation}
and we have the identity
\begin{equation}\label{eq:tan}
\langle \psi | \calH_\alpha \psi \rangle = \int_{\R^3} \left[ \frac {1}{2m} \left( k^2  n_\uparrow(k) - \frac { \con}{ |k-p_\uparrow|^2}\right)   + \frac {1} 2 \left(   k^2 n_\downarrow(k) - \frac {\con}{|k-p_\downarrow|^2} \right) \right] \dI k - \frac{m+1}{2m} \con \alpha
\end{equation}
\end{theorem}

Since $\con$, $p_\uparrow$ and $p_\downarrow$ are  uniquely determined by the momentum densities via \eqref{l1n}, Eq.~\eqref{eq:tan} expresses the energy solely in terms of the momentum densities. The set of possible momentum densities arising from wave functions $\psi\in D(\calH_\alpha)$ is not known, however, and can be expected to depend in a complicated way on both $\alpha$ and $N$.

The contact $\con$ thus determines the asymptotic behavior of both $n_\uparrow(k)$ and $n_\downarrow(k)$, via $n_\uparrow(k) \approx n_\downarrow(k) \approx \con |k|^{-4}$ for large $|k|$. In fact, up to terms decaying faster than $|k|^{-5}$, we have for large $|k|$ 
\begin{equation}
n_\uparrow(k) + n_\downarrow(k) \approx \frac\con{|k|^2 |k-p_\uparrow|^2} +  \frac\con{|k|^2 |k-p_\downarrow|^2} \approx \frac \con {|k - P|^4}
\end{equation}
for $P= \frac 12(p_\uparrow + p_\downarrow) = \|\xi\|_{L^2(\R^{3N})}^{-2} \int_{\R^{3N}} k_1 |\hat \xi(k_1,\dots,k_N)|^2 \dI k_1 \dots \dI k_N$. Note also that due to the fact that  $\lim_{K\to\infty} \int_{|k|<K} (  |k|^{-2} - |k-p|^{-2} ) \dI k = 0 $ for any $p\in \R^3$, one can rewrite the identity \eqref{eq:tan} as 
 \begin{equation}
\langle \psi | \calH_\alpha \psi \rangle = \lim_{K\to\infty } \int_{|k|<K} \left[ \frac {k^2}{2m} \left(  n_\uparrow(k) - \frac { \con}{ |k|^4}\right)   + \frac {k^2} 2 \left(   n_\downarrow(k) - \frac {\con}{|k|^4} \right) \right] \dI k - \frac{m+1}{2m} \con \alpha
\end{equation}

For any stationary state, the contact $\con$  can be computed as the derivative of the energy with respect to $\alpha$, by the Feynman-Hellmann principle. In fact,  for {\em fixed} $\psi$ (and hence fixed $\xi$),  
\begin{equation}\label{fhp}
\frac\partial{\partial \alpha} \calF_\alpha(\psi) = \frac{m+1}{2m} \con
\end{equation}
Note that it is important to use the quadratic form formulation here, as the domain of $\calH_\alpha$ depends on $\alpha$ and hence $\psi$ cannot be fixed when taking the derivative of $\langle\psi | \calH_\alpha\psi\rangle$ with respect to $\alpha$. Note also the minus sign in front of the last term in \eqref{eq:tan}; a naive derivative of  \eqref{eq:tan} would give the wrong sign!

The $L^1$-property \eqref{l1n} claimed in Theorem~\ref{thm:tan} does not make use of the boundary condition \eqref{bc} satisfied by $\psi \in D(\calH_\alpha)$ and holds more generally, in fact. The identity \eqref{eq:tan} only holds for $\psi$ satisfying  \eqref{bc}, however; i.e., it  holds for all functions $\psi$ in the domain of $\calH_\alpha$. (As already mentioned in the beginning of this section, self-adjointness of  $\calH_\alpha$ on this domain is not actually needed here. In particular, Theorem~\ref{thm:tan} holds for all $m>0$.)

The equations \eqref{def:rg}, \eqref{l1n}, \eqref{eq:tan} and \eqref{fhp} can be interpreted as a rigorous formulation of the  Tan relations  introduced in \cite{Tan}. There is actually one more relation, a  virial type theorem. It is an immediate consequence of the relation $U_\lambda^{-1} \calH_\alpha U_\lambda = \lambda^2 \calH_{\lambda^{-1}\alpha}$ for scaling the variables by $\lambda > 0$ and we shall not discuss it further here.

The proof of Theorem~\ref{thm:tan} will be given in Section~\ref{sec:tanproof}.

\section{Preliminaries}\label{sec:prelim}

Before giving the proof of the results in the previous section, we collect here a few auxiliary facts that will be used in the proofs. 

\begin{lemma}\label{lem:nl1}
The operator $\sigma$ on $L^2(\R^3)$ with integral kernel
\begin{equation}
\sigma(s,t) = (s^2+1)^{(\beta-1)/4} (t^2+1)^{-(\beta+1)/4} \frac 1{s^2 + t^2 +\lambda s\cdot t+1} 
\end{equation}
is bounded for $-2<\lambda<2$ and $-2<\beta < 2$.
\end{lemma}

\begin{proof}
We use the Schur test in the form
\begin{equation}\label{rssn}
\| \sigma\| \leq \frac 12 \sup_s  h(s) \int_{\R^3} h(t)^{-1} \left( |\sigma(s,t)| + |\sigma(t,s)| \right) \dI t
\end{equation}
for any positive function $h$, which is a consequence of the Cauchy-Schwarz inequality. Since $|\lambda|<2$, a pointwise estimate of the kernel reduces the problem to the case $\lambda =0$. Choosing $h(t) = (t^2 +1 )^{\gamma}$ one easily checks that the right side of (\ref{rssn}) is finite if and only if  $(1+|\beta|)/4<\gamma<(5-|\beta|)/4$.
\end{proof}

In the special case $\beta=0$, Lemma~\ref{lem:nl1} can be used to show that, for some $c>0$,  $|\Toff(\xi)| \leq c (N-1)  \Tdiag(\xi)$ for all $\xi\in H^{1/2}_{\rm as}(\R^{3(N-1)})$.
In particular, $F_\alpha$ is well-defined on its domain \eqref{def:DF}. 
Similarly, $\| L^{(\beta -1)/2} \Gamma \xi\|_{L^{2}(\R^{3 (N-1)})}$ is finite for $\xi \in H^{(\beta+1)/2}_{\rm as}(\R^{3 (N-1)})$ for $0\leq \beta < 2$. For $\beta=1$, this implies that the domain of $\Gamma$ contains $ H^{1}_{\rm as}(\R^{3 (N-1)})$.

\begin{lemma}\label{lem:nl2}
The operator $\sigma$ on $L^2(\R^3)$ with integral kernel
\begin{equation}\label{sig:lem2}
\sigma(s,t) = \left(  \frac { (s^2+\nu)^{(\beta-1)/4}}{ (t^2+\nu)^{(\beta+1)/4}} +   \frac { (t^2+\nu)^{(\beta-1)/4}}{ (s^2+\nu)^{(\beta+1)/4}}  \right)  \frac 1{s^2 + t^2 + \lambda s\cdot t +1} 
\end{equation}
is bounded and non-negative  for $-2 <  \beta <  2$, $\nu \geq 1/2$ and $-2<\lambda\leq 0$. 
\end{lemma}

\begin{proof}
Boundedness follows immediately from Lemma~\ref{lem:nl1}.  For $\beta=0$, positivity can be deduced   
from the integral representation
\begin{equation}\label{34}
 \left(  t^2 + s^2 +  \lambda  s \cdot t + 1 \right)^{-1} = \int_0^\infty  e^{-r  (1 + \lambda/2) t^2} e^{- r(1+ \lambda/2) s^2} e^{r \lambda (t - s)^2/2 } e^{-r}\dI r \,,
\end{equation}
noting that $-2< \lambda\leq 0$ and that the Gaussian has a positive Fourier transform. We are thus left with proving positivity for $\beta \neq 0$. Without loss of generality, we may assume $\beta > 0$, since $\sigma$ is invariant under the transformation $\beta\to -\beta$. To this aim, we use
\begin{equation}
x^{-\beta/2} = c_\beta \int_{0}^\infty \frac 1{x + r }  r^{-\beta/2} \dI r 
\end{equation}
with $c_\beta =\pi^{-1} \sin\left(\frac \pi 2 \beta\right)$  for $x>0$ and $0< \beta < 2$ to
rewrite the kernel as 
\begin{equation}
\sigma(s,t) = c_\beta  (s^2+\nu)^{(\beta-1)/4} (t^2+\nu)^{(\beta-1)/4} \int_{0}^\infty \left( \frac 1{s^2 + \nu + r }  + \frac 1{t^2+\nu + r} \right)  \frac {r^{-\beta/2}}{s^2 + t^2 + \lambda s\cdot t +1}  \dI r  
\end{equation}
Let us rewrite the integrand further as 
\begin{align}\nonumber
& r^{-\beta/2} \frac 1{s^2 + \nu + r }  \frac 1{t^2+\nu + r} \frac { s^2 + t^2 + 2 (\nu + r)  }{s^2 + t^2 + \lambda s\cdot t +1} \\
& = r^{-\beta/2} \frac 1{s^2 + \nu + r }  \frac 1{t^2+\nu + r} \left( 1 +  \frac {   2 ( \nu + r) - 1 -\lambda s\cdot t  }{s^2 + t^2 + \lambda s\cdot t +1} \right)  \label{seeth}
\end{align}
Using again \eqref{34}, as well as  $2(\nu+r) \geq 1$ and $\lambda\leq 0$, we see that \eqref{seeth} defines a non-negative operator. This completes the proof.
\end{proof}

\begin{lemma}\label{lem:nl3}
Consider the bounded operator $\sigma$ on $L^2(\R^3)$ with integral kernel given by \eqref{sig:lem2}
for $-2 <  \beta <  2$, $\nu \geq 1/2$ and $0\leq \lambda < 2$. 
Its positive and negative parts are the operators with kernels
\begin{align}\nonumber 
\sigma_+(s,t) & = \frac 12 \left( \sigma(s,t) + \sigma(s,-t) \right) \\
\sigma_-(s,t) & = -\frac 12 \left( \sigma(s,t) - \sigma(s,-t) \right)
\end{align}
respectively. 
\end{lemma}

\begin{proof}
Let $R$ denote the reflection operator $(R \varphi)(s) = \varphi(-s)$ for $\varphi\in L^2(\R^3)$. The operators $R$ and $\sigma$ clearly commute. Moreover, the product $\sigma R$ equals the operator with integral kernel \eqref{sig:lem2} and $\lambda$ replaced by $-\lambda$, which was shown to be non-negative in Lemma~\ref{lem:nl2}. One readily checks that this implies that the positive and negative parts of $\sigma$ are given by 
\begin{equation}
\sigma_\pm = \pm \frac 12  \sigma \left( 1 \pm R \right)\,,
\end{equation}
respectively. In fact, clearly $\sigma_+ \sigma_- = \sigma_- \sigma_+ = 0$, and $\sigma_\pm = \frac 12 \sigma R (1 \pm R)$, which is a product of two commuting nonnegative operators.
\end{proof}

\section{Proof of Theorem~\ref{thm1}}\label{sec:proof}

We assume $N\geq 3$ and define, for fixed  $\vec q \in \R^{3(N-2)}$ and $-2<\beta<2$,  an operator $\tau^\beta$ on $L^2(\R^{3})$ via the quadratic form
\begin{equation}\label{def:tau}
\langle \varphi | \tau^\beta | \varphi\rangle =  \frac 12  \int_{\R^{6}} \varphi^\ast (s )  \varphi (t ) \left( \frac {L(s,\vec q)^{(\beta-1)/2}}{ L(t,\vec q)^{(\beta+1)/2}} + \frac {L(t,\vec q)^{(\beta-1)/2}}{ L(s,\vec q)^{(\beta+1)/2}}  \right)  G(s,t,\vec q)  \dI s \dI t 
\end{equation}
where $L$ and $G$ are defined in \eqref{def:L} and \eqref{def:G}, respectively. 
Let $K := \sum_{i=1}^{N-2} q_i$, and recall that $A= 1/(m+2)$. 
The following observation is key to our further investigation. We shall need it here for $\beta=0$ only, but state it more generally for later use in the proof of Theorem~\ref{thm:domain}. 

\begin{lemma}\label{lem1}
The operator $\tau^\beta$ defined in \eqref{def:tau} is bounded on $L^2(\R^3)$. Its positive and negative parts, $\tau^\beta_\pm$, are the operators with integral kernels
\begin{align}\nonumber
\tau^\beta_+(s,t;\vec q) & =  \frac 14  \left( \frac {L(s,\vec q)^{(\beta-1)/2}}{ L(t,\vec q)^{(\beta+1)/2}} + \frac {L(t,\vec q)^{(\beta-1)/2}}{ L(s,\vec q)^{(\beta+1)/2}}  \right)    \left(   G(s,t,\vec q) +  G(s,-t- 2 AK,\vec q) \right)\\
\tau^\beta_-(s,t;\vec q ) & =   - \frac 14  \left( \frac {L(s,\vec q)^{(\beta-1)/2}}{ L(t,\vec q)^{(\beta+1)/2}} + \frac {L(t,\vec q)^{(\beta-1)/2}}{ L(s,\vec q)^{(\beta+1)/2}}  \right)   \left(   G(s,t,\vec q) -   G(s,-t- 2 AK,\vec q) \right) \label{def:taupm}
\end{align}
respectively.
\end{lemma}

\begin{proof}
Let  $Q^2 := \sum_{i=1}^{N-2} q_i^2$, and define  $\lambda:=2/(m+1)$. A simple calculation shows that
\begin{equation}\label{43}
G(s-A K,t-A K,\vec q)^{-1} 
=  t^2 + s^2 +  \lambda  s \cdot t + C
\end{equation}
where
\begin{equation}\label{44}
C = C(\vec q) =  \frac m{m+1} \left( A  K^2 + Q^2 \right) +\mu
\end{equation}
Similarly,
\begin{equation}\label{45}
L(s-AK,\vec q) = 2\pi^2 \left( \frac {m(m+2)}{(m+1)^2} s^2  +C \right)^{1/2}
\end{equation}
In particular, after a unitary translation by $AK$, the operator $\tau^\beta$ becomes the operator $\sigma$ with integral kernel
\begin{align}\nonumber
\sigma(s,t) & =  \frac{m+1}{4\pi^2} \left(   \frac{ \left[ {m(m+2)}  s^2  + (m+1)^2 C  \right]^{(\beta-1)/4} }{ \left[ {m(m+2)} t^2  + {(m+1)^2}   C \right]^{(\beta+1)/4}}  + \frac{ \left[ {m(m+2)}  t^2  + (m+1)^2 C  \right]^{(\beta-1)/4} }{ \left[ {m(m+2)} s^2  + {(m+1)^2}   C \right]^{(\beta+1)/4}} \right)  \\ &  \quad \times \left(  t^2 + s^2 +  \lambda  s \cdot t + C \right)^{-1}
\end{align} 

After a simple rescaling of the variables by $\sqrt{C}$, this is exactly of the form \eqref{sig:lem2}, with $\nu = (m+1)^2/(m(m+2))> 1/2$ (in fact, $>1$). Hence boundedness of $\sigma$ follows from Lemma~\ref{lem:nl1}. Moreover, Lemma~\ref{lem:nl3} applies, which states that 
 the positive and negative parts of $\sigma$ are given by 
\begin{equation}
\sigma_\pm = \pm \frac 12  \sigma \left( 1 \pm R \right)\,,
\end{equation}
where $R$ denotes reflection. 
  Undoing the unitary translation by $AK$, this leads to the statement of the lemma.
\end{proof}

For $\xi\in H^{1/2}_{\rm as}(\R^{3(N-1)})$, we define $\varphi \in L^{2}_{\rm as}(\R^{3(N-1)})$ by $ \varphi(s,\vec q) = L(s,\vec q)^{1/2} \hat\xi(s,\vec q)$. Then $\Tdiag(\xi) =  \| \varphi \|_{L^2(\R^{3(N-1)})}^2$, and 
\begin{align}\nonumber
 \Toff(\xi) &=  (N-1) \int_{\R^{3N}}  \varphi^\ast (s,\vec q )  \varphi (t,\vec q) L(s,\vec q)^{-1/2} L(t,\vec q)^{-1/2} G(s,t,\vec q)  \dI s \dI t \dI \vec q \\ &  \geq  - (N-1) \int_{\R^{3N}}  \varphi^\ast (s,\vec q)  \varphi (t,\vec q) \tau^0_-(s,t;\vec q)   \dI s \dI t \dI \vec q \label{tore}
\end{align}
where we simply dropped the positive part of the operator $\tau^0$ appearing on the right side. Its negative part, $\tau^0_-$, is explicitly identified in Lemma~\ref{lem1}. To proceed, we use the fact that $\varphi$ is antisymmetric. We introduce
\begin{equation}\label{def:tt}
\tilde\tau_-(s,\vec q,t,\vec \ell) = \tau^0_-(s,t;\vec q) \delta(\vec q - \vec \ell)
\end{equation}
for $\vec \ell \in \R^{3(N-2)}$, and rewrite the  term on the right side of  \eqref{tore} as
\begin{align}\nonumber
& (N-1) \int_{\R^{3N}}  \varphi^\ast (s,\vec q)  \varphi (t,\vec q) \tau^0_-(s,t;\vec q)   \dI s \dI t \dI \vec q \\ & = \sum_{i=0}^{N-2}   \int_{\R^{6(N-1)}}  \varphi^\ast (s,\vec q)  \varphi (t,\vec \ell)   \tilde\tau_-(q_i, \hat q_i ,\ell_i , \hat \ell_i)    \dI s \dI t \dI \vec q \dI \vec \ell \label{212}
\end{align}
where $\hat q_i = ( q_1, \dots, q_{i-1},s , q_{i+1} ,\dots, q_{N-2})$ and $\hat \ell_i = ( \ell_1, \dots, \ell_{i-1},t  , \ell_{i+1} ,\dots, \ell_{N-2})$ for $1\leq i \leq N-2$, as well as $q_0=s$, $\hat q_0 = \vec q$, $\ell_0 = t$, $\hat  \ell_0 = \vec \ell$.  To bound this last expression, we use the Schwarz inequality, as in \eqref{rssn}, to obtain
\begin{equation}
\eqref{212} \leq \| \varphi \|_{L^2(\R^{3(N-1)})}^2 \sup_{s,\vec q} h(s,\vec q)   \sum_{i=0}^{N-2}  \int_{\R^{3(N-1)}} h(t,\vec \ell)^{-1}  | \tilde\tau_-(q_i,\hat q_i,\ell_i, \hat \ell_i)|   \dI t  \dI \vec\ell  \label{tle}
\end{equation}
for any positive function $h$. Assume that $h$ is symmetric with respect to permutations. Inserting the special structure \eqref{def:tt}, the expression on the right side of  \eqref{tle} then equals
\begin{equation}\label{214}
 \| \varphi \|_{L^2(\R^{3(N-1)})}^2 \sup_{s,\vec q} h(s,\vec q)  \sum_{i=0}^{N-2} \int_{\R^{3}}   h(t,\hat q_i)^{-1}| \tau^0_-(q_i,t; \hat q_i) |  \dI t 
\end{equation}

We shall choose $h(s,\vec q) = s^2 \prod_{j=1}^{N-2} q_j^2$ in \eqref{214}. The resulting bound is then
\begin{align}\nonumber
\eqref{212} & \leq  \| \varphi \|_{L^2(\R^{3(N-1)})}^2 \sup_{s,\vec q}  \sum_{i=0}^{N-2}  \int_{\R^{3}}  \frac {q_i^2}{t^2} | \tau^0_-(q_i,t; \hat q_i) |   \dI t \\ & \leq \| \varphi \|_{L^2(\R^{3(N-1)})}^2 \sup_{s,\vec q}  \left( s^2 + Q^2 \right) \max_{0\leq i\leq N-2} 
 \int_{\R^{3}}  \frac {1}{t^2} | \tau^0_-(q_i,t; \hat q_i) |   \dI t 
\end{align}
where we again use the notation $Q^2 = \sum_{i=1}^{N-2} q_i^2$, as in the proof of Lemma~\ref{lem1}. Since for any $1\leq i \leq N-2$, $s^2 + Q^2$ is symmetric under exchange of $s$ and $q_i$, we can drop the maximum over $i$ when taking the supremum over $s$ and $\vec q$, and simply take $i=0$ (or any other value of $i$, in fact). We thus arrive at
\begin{equation}\label{rso}
\eqref{212}  \leq \| \varphi \|_{L^2(\R^{3(N-1)})}^2  \sup_{s,\vec q}  \left( s^2 + Q^2 \right)  \int_{\R^{3}}  \frac {1}{t^2} | \tau^0_-(s,t; \vec q) |   \dI t 
\end{equation}

To complete the proof of  \eqref{eq:thm}, we need to show that the term multiplying $ \| \varphi \|_{L^2(\R^{3(N-1)})}^2 = \Tdiag(\xi)$ on the right side of \eqref{rso} is bounded by $\Lambda(m)$. Recall the explicit expression of $\tau^0_-(s,t;\vec q)$, given in \eqref{def:taupm} above. We have 
\begin{align}\nonumber
|\tau^0_-(s,t;\vec q) | & =   \frac 1{\pi^2 (1+m)} \left( \frac{m}{(m+1)^2} (s+K)^2 + \frac m{m+1} (s^2 + Q^2) + \mu \right)^{-1/4}  \\ \nonumber &  \quad \times \left( \frac{m}{(m+1)^2} (t+K)^2 + \frac m{m+1} (t^2 + Q^2) + \mu \right)^{-1/4}   \\ & \quad \times   \frac {\left| (s+AK)\cdot (t+AK) \right|}{ \left[ (s+AK)^2 + (t+AK)^2 + \frac m{1+m}( Q^2 +A K^2) +\mu\right]^2  - \left[ \frac 2{(1+m)} (s+AK)\cdot(t+AK)\right]^2}    \label{tau-}
\end{align}
For an upper bound, we can replace $\mu$ by $0$. Moreover, we can replace the supremum over $\vec q\in \R^{3(N-2)}$ by a supremum over all $Q>0$ and $K\in \R^3$. This yields \eqref{eq:thm}.

To complete the proof of Theorem~\ref{thm1}, we have to show that $F_\alpha$ is closed for $\Lambda(m)<1$. This was already proved in \cite[Thm.~2.1]{correggi2012stability}, we include the proof here for completeness. Given a sequence $u_n\in D(F_\alpha)$ with $\|u_n-u_m\|_{L^2(\R^{3N})}\to 0$ and $F_\alpha(u_n-u_m)\to 0$ as $n,m\to \infty$, we need to show that there exists a $u\in D(F_\alpha)$ with $\lim_{n\to \infty} \|u_n - u\|_{L^2(\R^{3N})} = 0$ and $\lim_{n\to \infty} F_\alpha(u_n -u) = 0$. We choose any $\mu>0$ for $\alpha\geq 0$, and $\mu > \alpha^2 (2\pi(1-\Lambda(m))^{-2}$ for $\alpha<0$.  For such a choice, writing $u_n = w_n + G\xi_n$, the bound \eqref{eq:thm} implies that $\|w_n-w_m\|_{H^1(\R^{3N})} \to 0$ and $\| \xi_n - \xi_m \|_{H^{1/2}(\R^{3(N-1)})} \to 0$ as $n,m\to \infty$, and hence $w_n\to w$ and $\xi_n\to \xi$ for some $w$ and $\xi$, respectively, in the corresponding norms. Since $\|G(\xi_n-\xi_m)\|_{L^2(\R^{3N})} \leq \const \|\xi_n-\xi_m\|_{L^2(\R^{3(N-1)})}$, $u_n$ converges to $u=w+G\xi$ in $L^2(\R^{3N})$. Moreover, since $|F_\alpha(u_n-u)|$ is bounded from above by $\const( \|w_n-w\|_{H^1(\R^{3N})}^2 + \|\xi_n-\xi\|_{H^{1/2}(\R^{3(N-1)})}^2)$ (compare with the remark after Lemma~\ref{lem:nl1} in Section~\ref{sec:prelim}), the result follows. 
\hfill\qed

\bigskip

\begin{remark}
It is worth pointing out that the antisymmetry of the wave functions enters our proof of stability in three different ways. The first two concern the very definition of the model. First, there are no point interactions among the $N$ particles of mass $1$ themselves, due to the antisymmetry which forces the wave functions to vanish at particle coincidences. Second, the term $\Toff$ in the definition (\ref{def:Fal}) of the quadratic form $F_\alpha$ enters with a plus sign, while it would have a minus sign for bosons. This fact is crucial, as it allows to work with the negative part of the operator $\tau^0$ in \eqref{def:tau} instead of the positive part, which is larger. And third, we use the symmetry 
to replace the factor $(N-1)$ by a sum over particles in \eqref{212}.

 This last step would also work for bosons, only the symmetry of the absolute value of the wave functions is important. For the first two points, however, the antisymmetry is crucial. In the bosonic case, there is instability for any $N\geq 2$ and any $0<m<\infty$ \cite{Braaten06,Tamura1991,Yafaev74} (a fact known as the Thomas effect \cite{Thomas1935}).  
  While $\Toff$ can be bounded from below by $-\Tdiag$, as Theorem~\ref{thm1} shows, it  is in fact known that $\Toff(\xi)\leq \Tdiag(\xi)$ is false for suitable $\xi$ for any $m$ \cite{correggi2012stability}.
\end{remark}

\section{Proof of Theorem~\ref{thm:domain}}\label{sec:domain}

Let us define the operator $J$ by $\Gamma = L + J$, i.e., $\Toff(\xi) = \langle \xi | J \xi\rangle$ for $\xi \in H^1_{\rm as}(\R^{3(N-1)})$. For $0\leq \beta <2$, we have 
\begin{align}\nonumber
\| L^{(\beta -1)/2} \Gamma \xi\|_{L^{2}(\R^{3 (N-1)})}^2 & = \| L^{(\beta + 1)/2} \xi\|_{L^{2}(\R^{3 (N-1)})}^2 + \langle \xi | ( J L^{\beta} + L^\beta J ) \xi\rangle + \| L^{(\beta -1)/2} J \xi\|_{L^{2}(\R^{3 (N-1)})}^2
\\ & \geq \| L^{(\beta + 1)/2} \xi\|_{L^{2}(\R^{3 (N-1)})}^2 + \langle \xi | ( J L^{\beta} + L^\beta J ) \xi\rangle
\end{align}
for all $\xi \in H^{(\beta+1)/2}_{\rm as}(\R^{3 (N-1)})$. The result \eqref{217} thus follows if we can show that 
\begin{equation}
\langle \xi | ( J L^{\beta} + L^\beta J ) \xi\rangle   \geq  - \Lambda_\beta(m)  \| L^{(\beta+1)/2} \xi \|^2_{L^{2}(\R^{3 (N-1)})}
\end{equation} 
With $\varphi = L^{(\beta+1)/2} \xi$ this reads, equivalently, 
\begin{equation}
\langle \varphi | ( L^{-(\beta+1)/2} J L^{(\beta-1)/2} + L^{(\beta-1)/2} J  L^{-(\beta+1)/2} )\varphi \rangle   \geq  - \Lambda_\beta(m)  \| \varphi \|^2_{L^{2}(\R^{3 (N-1)})}
\end{equation} 
for all $\varphi \in L^{2}_{\rm as}(\R^{3 (N-1)})$. The left side equals
\begin{equation}\label{54}
(N-1)\int_{\R^{3N}} \hat\varphi^\ast (s, \vec q)  \hat\varphi (t, \vec q ) \left(  \frac{L(t,\vec q)^{(\beta-1)/2}}{ L(s,\vec q)^{(\beta+1)/2}}  + \frac{ L(s,\vec q)^{(\beta-1)/2} }{L(t,\vec q)^{(\beta+1)/2}} \right)  G(s,t,\vec q)  \dI s \dI t  \dI \vec q
\end{equation}
where $\vec q \in \R^{3(N-2)}$ and  $L$ and $G$ are defined in \eqref{def:L} and \eqref{def:G}, respectively.

The above integral over $s$ and $t$, for fixed $\vec q$, is the expectation of (twice) the operator $\tau^\beta$ defined in \eqref{def:tau}. Lemma~\ref{lem1} identifies its negative and positive parts. Dropping the latter, we thus have 
\begin{align}\nonumber
\eqref{54} & \geq   (N-1)\int_{\R^{3N}} \hat\varphi^\ast (s, \vec q)  \hat\varphi (t, \vec q ) \left(  \frac{L(t,\vec q)^{(\beta-1)/2}}{ L(s,\vec q)^{(\beta+1)/2}}  + \frac{ L(s,\vec q)^{(\beta-1)/2} }{L(t,\vec q)^{(\beta+1)/2}} \right) \\  & \qquad\qquad\qquad \times \frac 12 \left(  G(s,t,\vec q) - G(s,-t-2AK,\vec q)\right)  \dI s \dI t  \dI \vec q
\end{align}
The remainder of the proof proceeds in exactly the same way as in the proof of Theorem~\ref{thm1}, Eqs.~\eqref{def:tt}--\eqref{rso}, and we shall not repeat it here. The result is \eqref{217}, for any $0\leq \beta< 2$. The limiting case $\beta=2$ is then obtained by monotone convergence, using that $\Lambda_\beta(m)$ is convex and thus continuous in $\beta$. (Note that for $\beta=2$, the left side of \eqref{217} need not  be finite, a priori.) \hfill\qed

\section{Upper Bound on $\Lambda_\beta(m)$}\label{sec:Lambda}

In this section we shall prove an upper bound on $\Lambda_\beta(m)$. While only the case $0\leq \beta\leq 2$ is of interest here, our bound is actually valid for all $0\leq \beta<3$. We start with proving the bound \eqref{bound:L} on $\Lambda(m)$. Recall the definitions of $\Lambda(m)$ and $\ell_m$ in \eqref{def:Lm} and \eqref{def:lm}, respectively, as well as $A=(2+m)^{-1}$. We shall use that
\begin{equation}\label{61}
\ell_m(s,K,Q) \geq \frac {\sqrt{m(m+2)}}{m+1} | s + AK | 
\end{equation}
and that
\begin{align}\nonumber 
&  \left[ (s+AK)^2 + (t+AK)^2 + \frac m{1+m}( Q^2 +A K^2) \right]^2  - \left[ \frac 2{(1+m)} (s+AK)\cdot(t+AK)\right]^2 \\  \nonumber
& \geq \frac{m(m+2)} { (1+m)^2}   \left[ (s+AK)^2 + (t+AK)^2 + \frac m{1+m}( Q^2 +A K^2) \right]^2  \\ 
& \geq \frac{m(m+2)} { (1+m)^2}   \left[  \frac{m(2+m)}{2+4m+m^2} (s^2 + t^2) + \frac m{1+m} Q^2   \right]^2 
\end{align}
Together with the simple bound
\begin{equation}\label{63}
|s+AK|^{1/2} |t+AK|^{1/2} \leq \sqrt{ \frac 12 (s+AK)^2 + \frac 12 (t+AK)^2}
\end{equation}
this gives
\begin{align}\nonumber
\Lambda(m)& \leq  \frac{ (1+m)^2 }{\sqrt{2} \pi^2  \left[ m(m+2) \right]^{3/2} }   \sup_{s\in\R^3, Q>0}  \int_{\R^3}  \frac 1{t^2}    \frac { s^2 + Q^2 }{  \left[  \frac{m(2+m)}{2+4m+m^2} (s^2 + t^2) + \frac m{1+m} Q^2   \right]^{3/2}   }    \dI t \\
 & =  \frac{ 4 (1+m)^2 (2+4m +m^2)^{3/2}}{\sqrt{2} \pi  \left[ m(m+2) \right]^{3} }  \sup_{s\in\R^3, Q>0}   \frac { s^2 + Q^2 }{   s^2  +    \frac{2+4m+m^2}{(2+m)(1+m)} Q^2}   
 \end{align}
Since $2+4m+m^2> (2+m)(1+m)$, the last supremum equals $1$, and we obtain the bound \eqref{bound:L}.

The same strategy can be used to derive an upper bound on $\Lambda_\beta(m)$ in \eqref{def:Lmb}, for $\beta \leq 1$. Instead of \eqref{63}, one uses
\begin{equation}\label{65}
|s+AK|^{(1+\beta)/2} |t+AK|^{(1-\beta)/2} + |s+AK|^{(1-\beta)/2} |t+AK|^{(1+\beta)/2} \leq \sqrt{ 2 (s+AK)^2 + 2 (t+AK)^2}
\end{equation}
(which follows from convexity of the exponential function, $xy \leq  \frac 1p x^p + \frac 1q y^q$ for $x,y\geq 0$, $p > 1$, $\frac 1p + \frac 1q = 1$),  resulting in 
\begin{equation}
\Lambda_\beta (m) \leq    \frac{ 4 \sqrt{2} (1+m)^2 (2+4m +m^2)^{3/2}}{\pi  \left[ m(m+2) \right]^{3} } \quad \text{for $\beta \leq 1$.}
 \end{equation}
For $1<\beta<3$, we need an upper bound on $\ell_m$, and we shall simply use
\begin{equation}
\ell_m(s,K,Q) \leq \sqrt{ (s+AK)^2 + \frac m{m+1} (Q^2 + AK^2)}
\end{equation}
For a lower bound, we shall use \eqref{61} for one power of $\ell_m$, and 
\begin{equation}
\ell_m(s,K,Q) \geq \sqrt{ \frac m{m+1} (s^2 + Q^2)}
\end{equation}
for the remaining $\ell_m^{(\beta-1)/2}$. This leads to
\begin{align}\nonumber
\Lambda_\beta(m)& \leq  \frac{1}{\pi^2} \frac {(m+1)^{(\beta+7)/4}}{m^{(\beta+5)/4}(2+m)^{3/2}}    \sup_{s\in\R^3, Q>0}  \int_{\R^3}  \frac 1{t^2} \left( \frac 1{|t|^{(\beta-1)/2} } + \frac 1 {(s^2+Q^2)^{(\beta-1)/4}}  \right) \\ \nonumber & \qquad\qquad \qquad\qquad \qquad\qquad\qquad \times  \frac { s^2 + Q^2 }{  \left[  \frac{m(2+m)}{2+4m+m^2} (s^2 + t^2) + \frac m{1+m} Q^2   \right]^{(7-\beta)/4}   }    \dI t \\
 & =   \frac{4}{\pi} \frac {(m+1)^{(\beta+7)/4}}{m^{3}(2+m)^{(13-\beta)/4}}   \left( {2+4m+m^2} \right)^{(7-\beta)/4}   \left( \frac{2}{3-\beta} +  \frac{\sqrt{\pi}}{2} \frac{\Gamma((5-\beta)/4)}{ \Gamma((7-\beta)/4)} \right)  
 \end{align}
 for $1 < \beta < 3$, where $\Gamma$ denotes the gamma-function in the last expression. In particular, $\Lambda_\beta(m)$ is finite for $\beta <3$, and decays at least like $m^{-1}$ for large $m$.

\section{Numerical Evaluation of $\Lambda_\beta(m)$}\label{sec:numerics}

Recall the definition of $\Lambda(m)$ in \eqref{def:Lm}. In order to obtain a numerical value for $\Lambda(m)$, it is convenient to simplify this expression a bit. 
As a first step, we claim that, given $s$, the supremum over $K$ in \eqref{def:Lm} is attained at some $K$ of the form $K=- bs $ for $0\leq b \leq  1/A=2+m$. To see this, we substitute $\tilde s = s+AK$, $\tilde t = t + AK$, and rewrite \eqref{def:Lm} as
\begin{align}\nonumber 
\Lambda(m) & =  \sup_{\tilde s,K\in\R^3, Q>0}  \frac{(\tilde s-AK)^2 + Q^2}{\pi^2 (1+m)} \left( \frac{m(m+2)}{(m+1)^2} \tilde s^2 + \frac m{m+1} ( Q^2 + A K^2)  \right)^{-1/4} \\ \nonumber & \quad \quad  \times  \int_{\R^3}  \frac 1{(\tilde t-AK)^2}     \left( \frac{m(m+2)}{(m+1)^2} \tilde t^2 + \frac m{m+1} (Q^2+AK^2)  \right)^{-1/4}   \\ & \qquad\qquad \times   \frac { \left|\tilde s\cdot \tilde t\right|  }{ \left[ \tilde s^2 + \tilde t^2 + \frac m{1+m}( Q^2 +A K^2) \right]^2  - \left[ \frac 2{(1+m)} \tilde s\cdot \tilde t\right]^2}    \dI \tilde t \label{320}
\end{align}
Since the term on the last line is invariant under the reflection $\tilde t \mapsto - \tilde t$, the integral above is equal to 
\begin{align}\nonumber 
 &   \int_{\R^3}  \frac {\tilde t^2 + A^2 K^2}{(\tilde t^2+A^2K^2)^2 - 4A^2 (\tilde t\cdot K)^2}     \left( \frac{m(m+2)}{(m+1)^2} \tilde t^2 + \frac m{m+1} (Q^2+AK^2)  \right)^{-1/4}   \\ & \qquad\qquad \times \frac { \left|  \tilde s\cdot \tilde t \right|}{ \left[ \tilde s^2 + \tilde t^2 + \frac m{1+m}( Q^2 +A K^2) \right]^2  - \left[ \frac 2{(1+m)} \tilde s\cdot \tilde t\right]^2}     \dI \tilde t \label{321}
\end{align}
When optimizing over the orientation of $\tilde s$ and $K$, the very first factor after the supremum in \eqref{320} is clearly largest if $\tilde s$ and $K$ are antiparallel. That the same is true for the integral \eqref{321} is the content of the following lemma, whose proof is an easy exercise.

\begin{lemma}\label{lem2}
Let $f$ and $g$ be measurable  functions on $[-1,1]$ that are non-negative, even, and increasing on $[0,1]$. For $a,b\in \Ss^2$,  
\begin{equation}\label{intsph}
\int_{\Ss^2} f(\omega\cdot a) g(\omega\cdot b) \dI \omega
\end{equation}
is largest if $a$ and $b$ are either parallel or antiparallel (as vectors in $\R^3$).
\end{lemma}

\begin{proof} We can represent the functions $f$ and $g$ by their level sets, and write
\begin{equation}\label{lsi}
\eqref{intsph} =  \int_{\Ss^2\times \R_+^2} \chi_{\{f > x \} }(\omega\cdot a) \chi_{\{g > y\}}(\omega\cdot b) \dI \omega \dI x \dI y
\end{equation}
The support of the function $\omega \mapsto \chi_{\{f > x \} }(\omega\cdot a)$ consists of  the union of two spherical caps, centered at $\pm a$, respectively, and similarly for $\chi_{\{g > y\}}(\omega\cdot b) $.  If $\pm a$ is parallel to $b$, the integral over $\Ss^2$ in \eqref{lsi} (for fixed $x$ and $y$) is clearly largest, since one of the characteristic functions simply equals $1$ on the support of the other in this case. This completes the proof.
\end{proof}

The angular part of the integral in \eqref{321} is exactly of the form \eqref{intsph}. We thus 
 conclude that we can restrict the supremum in \eqref{320} to the set where $K = - \kappa \tilde s$ for some $\kappa \geq 0$ or, equivalently, $K = - b s$ for some $0\leq b=\kappa/(1+\kappa A)\leq 1/A$.

To evaluate $\Lambda(m)$, we thus have to find the supremum over $\tilde s\in \R^3$, $\kappa\geq 0$ and $Q\geq 0$ of 
\begin{align}\nonumber &
  \frac{\tilde s^2 (1+\kappa A)^2 + Q^2}{\pi^2 (1+m)} \left( \frac{m(m+2)}{(m+1)^2} \tilde s^2 + \frac m{m+1} ( Q^2 + A \kappa^2 \tilde s^2)  \right)^{-1/4} \\ \nonumber & \quad \quad  \times  \int_{\R^3}   \frac {\tilde t^2 + A^2 \kappa^2 \tilde s^2}{(\tilde t^2+A^2 \kappa^2 \tilde s^2)^2 - 4A^2 \kappa^2 (\tilde t\cdot \tilde s)^2}     \left( \frac{m(m+2)}{(m+1)^2} \tilde t^2 + \frac m{m+1} (Q^2+A \kappa^2 \tilde s^2)  \right)^{-1/4}   \\ & \qquad\qquad \times  \frac {  \left|\tilde s\cdot \tilde t  \right|}{ \left[ \tilde s^2 + \tilde t^2 + \frac m{1+m}( Q^2 +A \kappa^2 \tilde s^2) \right]^2  - \left[ \frac 2{(1+m)} \tilde s\cdot \tilde t\right]^2}    \dI \tilde t \label{3201}
\end{align}
After carrying out the angle integration, this becomes
\begin{align}\nonumber &
 2  \frac{\tilde s^2 (1+\kappa A)^2 + Q^2}{\pi (1+m)} \left( \frac{m(m+2)}{(m+1)^2} \tilde s^2 + \frac m{m+1} ( Q^2 + A \kappa^2 \tilde s^2)  \right)^{-1/4} \\ \nonumber & \quad \quad  \times  \int_{0}^\infty    \frac { t^2}{t^2+A^2 \kappa^2 \tilde s^2}     \left( \frac{m(m+2)}{(m+1)^2}  t^2 + \frac m{m+1} (Q^2+A \kappa^2 \tilde s^2)  \right)^{-1/4}   \\ & \qquad\qquad \times  \frac {|\tilde s| t}{ \left[ \tilde s^2 +  t^2 + \frac m{1+m}( Q^2 +A \kappa^2 \tilde s^2) \right]^2  }  \frac {\ln (1-\lambda_1) - \ln(1-\lambda_2)}{\lambda_2-\lambda_1}    \dI  t \label{3202}
\end{align}
where
\begin{equation}
\lambda_1 = \frac{ 4 A^2 \kappa^2 t^2 \tilde s^2}{(t^2 + A^2\kappa^2 \tilde s^2)^2} \quad , \qquad \lambda_2 = \frac 4{(m+1)^2} \frac{t^2 \tilde s^2}{(t^2 + \tilde s^2 + \frac m{m+1} ( Q^2 + A\kappa^2 \tilde s^2))^2} 
\end{equation}

By the overall scale invariance, we can set $\tilde s^2 = 1$, and hence we are left with two parameters to optimize over, $Q\geq 0$ and $\kappa\geq 0$ or, equivalently, $0\leq b\leq 1/A=2+m$. It is not difficult to see that \eqref{3202} tends to zero as $Q\to \infty$ (uniformly in $b$) and thus the optimization is effectively over a compact set. The result of a numerical integration of \eqref{3202} in the case $m=1$ is shown in Figure~\ref{fig:maxPlotNum}. The supremum is attained at $Q=0$ and $b\approx 0.82$, and equals $\Lambda(1) \approx 0.34$. In particular, it is less than $1$. Moreover, the numerical evaluation yields $\Lambda(m)<1$ for all  $m\geq 0.36$, i.e., the critical mass for stability is less than $0.36$, as shown in Figure~\ref{fig:lambda1}.

The same analysis applies to $\Lambda_\beta(m)$ in \eqref{def:Lmb}. For $\beta=1$ and $\beta=2$, the graph of these functions is plotted in Figure~\ref{fig:lambda1}. 

\begin{figure}
\begin{center}
\includegraphics[width = 10cm]{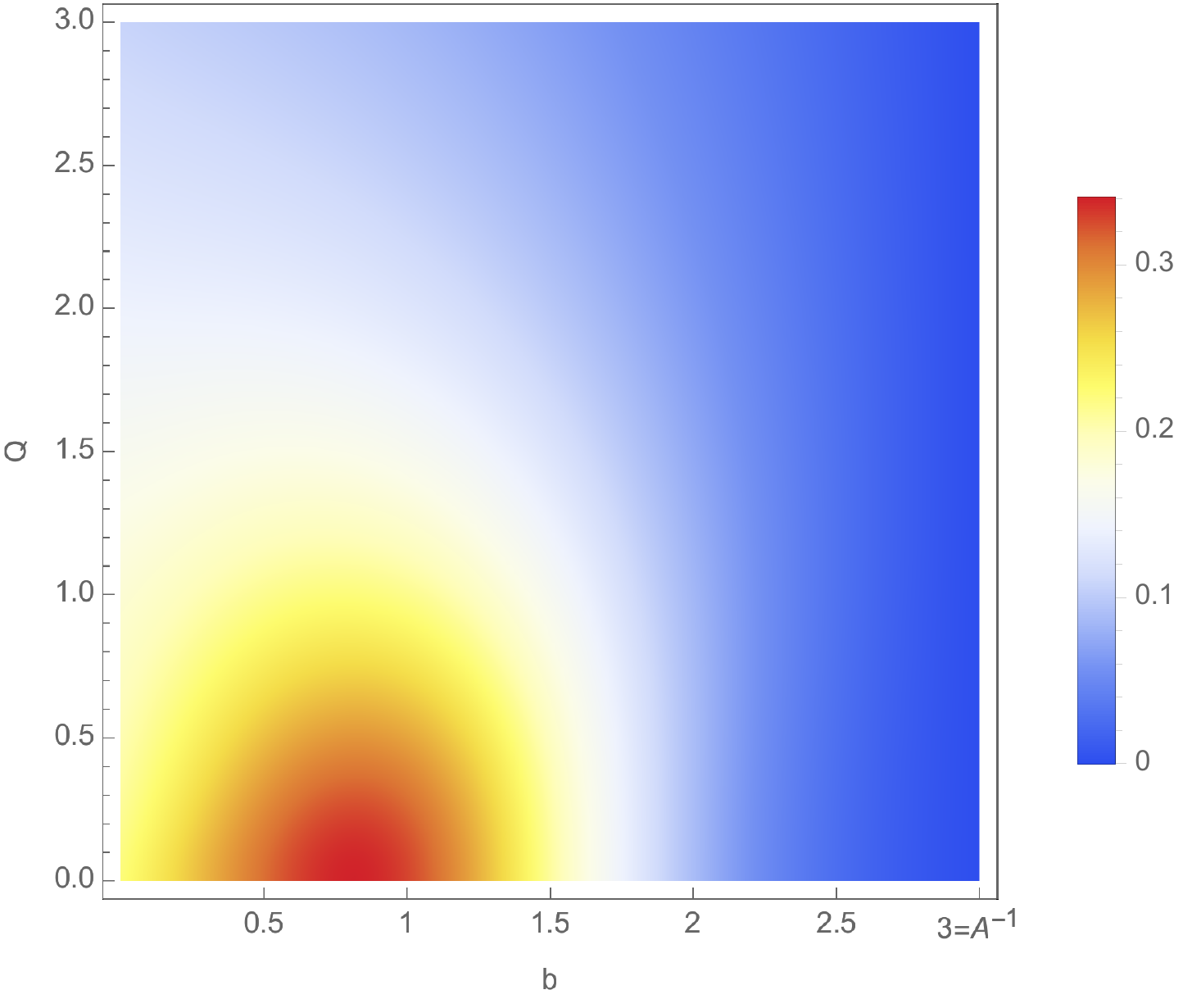}
\end{center}
\caption{Numerical evaluation of the expression  \eqref{3202} (for $\tilde s^2=1$), whose maximal value is $\Lambda(1)$. The maximum is attained at $Q=0$ and $b\approx 0.82$, and has a value $\Lambda(1) \approx 0.34$.}\label{fig:maxPlotNum}
\end{figure}

\section{Proof of Proposition~\ref{prop1}}\label{sec:propp}

Let $\psi \in D(\calH_\alpha)$, and consider 
 the partial Fourier transform
\begin{equation}
\eta(P, k_1,\ldots,k_{N-1} , r) = \frac 1{ (2\pi)^{3/2} } \int_{\R^3} \hat \psi\left( \tfrac{m}{1+m}P + q , k_1,\ldots,k_{N-1} , \tfrac{1}{1+m} P - q\right) e^{i r \cdot q} \dI q
\end{equation}
With the aid of \eqref{eq:tquadraticForm} and \eqref{def:calL}--\eqref{bc} we can write
\begin{align}\nonumber
\eta(P, k_1,\ldots,k_{N-1} , r) & =  \left( \frac {2\pi^2} {|r|} + \alpha\right) \frac{2m}{m+1}\frac{(-1)^{N+1}}{(2\pi)^{3/2}} \hat\xi(P,k_1,\ldots k_{N-1}) \\ & \quad + \sum_{j=1}^3 \kappa_j(P, k_1,\ldots,k_{N-1} , r)
\end{align}
where
\begin{equation}
\kappa_1(P, k_1,\ldots,k_{N-1} , r) = \frac 1{ (2\pi)^{3/2} }  \int_{\R^3} \hat \phi\left( \tfrac{m}{1+m}P + q , k_1,\ldots,k_{N-1} , \tfrac{1}{1+m} P - q\right) \left( e^{i r \cdot q} - 1 \right) \dI q 
\end{equation}
and 
\begin{align}\nonumber
& \kappa_2(P, k_1,\ldots,k_{N-1} , r) \\ \nonumber &= \frac 1{ (2\pi)^{3/2} }   \int_{\R^3} \calG\left( \tfrac{m}{1+m}P + q , k_1,\ldots,k_{N-1} , \tfrac{1}{1+m} P - q\right) \left(  e^{i r \cdot q}-1 \right)\\ & \qquad\qquad\quad \times  \sum_{j=1}^{N-1} (-1)^{j+1} \hat \xi\left( \tfrac{m}{1+m}P + q + k_j  , k_1,\ldots,k_{j-1}, k_{j+1} \dots,k_{N-1} , \tfrac{1}{1+m} P - q\right)    \dI q
\label{def:kappa2}
\end{align}
Introducing the function $f(t) = t^{-1} (e^{-t} - 1 + t)$ for $t>0$ we further have 
\begin{align}\nonumber
& \kappa_3(P, k_1,\ldots,k_{N-1} , r) \\ & = \frac{(-1)^{N+1}}{(2\pi)^{3/2}}  f \left( \frac {|r|}{2\pi^2} \frac{1+m}{2m} \calL(P,k_1,\ldots,k_{N-1}) \right)  \calL(P,k_1,\ldots,k_{N-1})  \hat\xi(P,k_1,\ldots k_{N-1})
\end{align}

Since $\phi\in H^2(\R^{3(N+1)})$, one readily checks that $\lim_{r\to 0} \| \kappa_1(\,\cdot\, ,r)\|_{L^2(\R^{3N})} = 0$. Moreover, since $\xi\in H^{1}(\R^{3N})$ by assumption, $\lim_{r\to 0} \| \kappa_3(\,\cdot\, ,r)\|_{L^{2}(\R^{3N})} = 0$ by dominated convergence, using $\lim_{t\to 0}f(t)=0$. The same holds true for $\kappa_2$ if we can show that 
\begin{equation}\label{tsl2}
\int_{\R^3} \calG\left( \tfrac{m}{1+m}P + q , k_1,\ldots,k_{N-1} , \tfrac{1}{1+m} P - q\right) \left|  \hat \xi\left( \tfrac{m}{1+m}P + q + k_1 , k_2, \ldots,k_{N-1} , \tfrac{1}{1+m} P - q\right)  \right|  \dI q 
\end{equation}
is an $L^2(\R^{3N})$ function. For this purpose,  
pick a function $\nu \in L^{2}(\R^3)\otimes L^{2}_{\rm as}(\R^{3(N-1)})$ and integrate the expression \eqref{tsl2} against $\nu(P,k_1,\ldots k_{N-1})$. After a change of integration variables, this gives 
\begin{equation}\label{inde}
\int_{\R^{3(N+1)}} \nu(k_0+k_N,k_1,\dots,k_{N-1}) \calG\left(k_0 , k_1,\ldots, k_N \right) \left|  \hat \xi\left( k_0 + k_1 , k_2, \ldots,k_N \right)  \right|  \dI k_0   \dI k_1 \cdots \dI k_{N} 
\end{equation}
Since $\xi \in H^1(\R^{3N})$ by assumption, Lemma~\ref{lem:nl1} (for $\beta=1$) implies that \eqref{inde} is finite. This shows that also $\| \kappa_2(\,\cdot\, ,r)\|_{L^{2}(\R^{3N})}$ goes to $0$ as  $r\to 0$, and thus completes the proof of  Proposition~\ref{prop1}. \hfill\qed

\section{Proof of Theorem~\ref{thm:tan}}\label{sec:tanproof}

We start with $n_\uparrow$. For $\psi = \phi + \calG \xi \in D(\cal H_\alpha)$, we have 
\begin{align}\nonumber
& k^2 n_\uparrow(k) - \frac{\con}{|k - p_\uparrow  |^2} 
\\ \nonumber &= k^2 \int_{\R^{3N}} | \hat \phi(k , k_1,k_2, \vec k)|^2 \dI k_1 \dI k_2 \dI \vec k
\\  \nonumber & \quad  - k^2 N(N-1) \int_{\R^{3N}}  \calG(k, k_1, k_2, \vec k)^2  \hat \xi^*(k+ k_1 , k_2,  \vec k)\hat\xi(k + k_2 , k_1,  \vec k )\dI k_1 \dI k_2 \dI \vec k 
\\  \nonumber & \quad  +  N  \int_{\R^{3N}} \left(  k^2 \calG(k, k_1, k_2, \vec k)^2 - \left( \frac {2m}{m+1}\right)^2 \frac {1}{|k-p_\uparrow|^2} \right)  |\hat \xi(k+k_1 ,  k_{2}, \vec k)|^2  \dI k_1 \dI k_2 \dI \vec k 
\\ & \quad  + 2 k^2 N \Re  \int_{\R^{3N}}  \hat \phi^*(k , k_1, k_2, \vec k)\calG(k, k_1, k_2, \vec k)  \hat \xi(k+ k_1 ,  k_{2},  \vec k)    \dI k_1 \dI k_2 \dI \vec k \label{nup}
\end{align}
where  $\vec k \in \R^{3(N-2)}$, as before. 
We write the right side as $\sum_{j=1}^4 M^\uparrow_j(k)$, with $M^\uparrow_j$ corresponding to the term on the $j$th line on the right side. The first term $M_1^\uparrow$ is clearly in $L^1(\R^3)$. Using \eqref{def:calG} the second term can be bounded as
\begin{equation}
|M_2^\uparrow(k)|   \leq    N(N-1) \int_{\R^{3N}}  \frac {4m}{k_1^2 + k_2^2 } | \hat\xi(k+ k_1 ,  k_{2}, \vec k) | | \hat\xi(k + k_2 , k_1, \vec k)| \dI k_1 \dI k_2 \dI \vec k 
\end{equation}
After integrating over $k$ and using the Cauchy-Schwarz inequality for the $(k, \vec k)$ integration, we get
\begin{align}\nonumber
\int_{\R^3} |M_2^\uparrow(k)|  \dI k  &  \leq    N(N-1) \int_{\R^{3N}}  \frac {4m}{k_1^2 + k_2^2 }  \| \hat\xi (\,\cdot\, ,k_1)\|_{L^2(\R^{3(N-1)})} \| \hat\xi (\,\cdot\, ,k_2)\|_{L^2(\R^{3(N-1)})}  \dI k_1 \dI k_2 \\& \leq  4m c N(N-1)^2  \| \xi \|_{H^{1/2}(\R^{3N})}^2 
\end{align}
where $c$ equals the norm of the operator with integral kernel $|k_1|^{-1/2} |k_2|^{-1/2} (k_1^2 + k_2^2)^{-1}$, which can easily be shown to be finite (and, in fact, equals $2\pi^2$ \cite[Lemma~2.1]{Finco2012}). 

Next we shall consider $M_3^\uparrow(k)$, which we rewrite as 
\begin{equation}\label{inse}
M_3^\uparrow(k) =  N  \int_{\R^{3N}} \left(  k^2 \calG(k, k_1-k, k_2, \vec k)^2 - \left( \frac {2m}{m+1}\right)^2 \frac {1}{|k-p_\uparrow|^2} \right)  |\hat \xi(k_1 ,  k_{2}, \vec k)|^2  \dI k_1 \dI k_2 \dI \vec k 
\end{equation}
Since $\xi \in L^2(\R^{3N})$, $M_3^\uparrow$ is clearly in $L^1_{\rm loc}(\R^3)$ and we only have to investigate its behavior for large $k$. If we write
\begin{equation}
k^2 \calG(k, k_1-k, k_2 , \vec k)^2 -  \left( \frac {2m}{m+1}\right)^2  \frac {1}{|k-p_\uparrow|^2} =  \left( \frac {2m}{m+1}\right)^2 \frac {2}{|k|^4} k \cdot \left( \frac {2m}{m+1} k_1 - p_\uparrow \right) + R_\uparrow(k, k_1, k_2, \vec k) 
\end{equation}
the first term on the right side gives zero after integration when inserted in \eqref{inse}, by the definition of $p_\uparrow$ in \eqref{defbark}. That is,
\begin{equation}\label{symfor}
M_3^\uparrow(k) =  N  \int_{\R^{3N}} R_\uparrow(k, k_1, k_2, \vec k)    |\hat \xi(k_1 ,  k_{2}, \vec k)|^2  \dI k_1 \dI k_2 \dI \vec k 
\end{equation}
Moreover, in the region  where $|k|^2 \geq \const ( \mu + p_\uparrow^2)$ we have 
\begin{equation}\label{readc}
| R_\uparrow (k, k_1, k_2, \ldots, k_N) | \leq \const   \frac 1{|k|^3}  \left( \mu + p_\uparrow^2 + \sum_{j=1}^N k_j^2 \right)^{1/2} \min \left\{ 1 , \frac 1{|k|} \left( \mu + p_\uparrow^2 + \sum_{j=1}^N k_j^2 \right)^{1/2} \right\}
\end{equation}
for suitable constants. If we integrate $R_\uparrow$ over $k$ in this region we thus obtain an expression that is bounded from above by $\const ( \mu + p_\uparrow^2 + \sum_{j=1}^N k_j^2 )^{1/2} \ln ( 1+  \mu + p_\uparrow^2 + \sum_{j=1}^N k_j^2)$, and we conclude, in particular, that $\| M_3^{\uparrow}\|_{L^1(\R^3)} \leq \const \| \xi\|_{H^{1}(\R^{3N})}^2$. 
Finally, using the simple pointwise bound 
\begin{equation}
|M_4^\uparrow(k) | \leq 4 m N \| \hat \phi(k, \,\cdot\,)\|_{L^2(\R^{3N})}  \| \xi\|_{L^2(\R^{3N})} 
\end{equation}
and the assumption that $\phi \in H^2(\R^{3(N+1)})$, the Cauchy-Schwarz inequality readily implies that $M_4^\uparrow \in L^1(\R^3)$. This concludes the proof that $ k^2 n_\uparrow(k) - {\con} |k - p_\uparrow |^{-2}$ is integrable.

Similarly we have for $n_\downarrow$
\begin{align}\nonumber
& k^2  n_\downarrow(k) - \frac{\con}{|k - p_\downarrow|^2}   = \sum_{j=1}^7 M_j^\downarrow(k) = 
\\ \nonumber &= N k^2 \int_{\R^{3N}} | \hat \phi(k_0 , k, k_2, \vec k)|^2 \dI k_0 \dI k_2  \dI \vec k 
\\  \nonumber & \quad  - k^2 N(N-1)(N-2) \int_{\R^{3N}}  \calG(k_0, k, k_2, \ldots, k_N)^2  \hat \xi^*(k_0+ k_2 ,  k , k_3, \ldots , k_N) 
\\ \nonumber & \qquad\qquad\qquad \qquad\qquad\qquad \times \hat\xi(k_0 + k_3 , k ,k_2, k_4,  \ldots, k_N)\dI k_0 \dI k_2 \cdots \dI k_N 
 \\  \nonumber & \quad - 2 k^2 N(N-1)  \int_{\R^{3N}}  \calG(k_0, k, k_2, \vec k)^2  \hat \xi^*(k_0+ k , k_2, \vec k)\hat\xi(k_0 + k_2 , k ,\vec k)\dI k_0 \dI k_2 \dI \vec k 
 \\  \nonumber & \quad  + k^2 N(N-1) \int_{\R^{3N}}  \calG(k_0, k, k_2, \vec k)^2  |\hat \xi(k_0+ k_2 ,  k ,\vec k)|^2 \dI k_0 \dI k_2 \dI \vec k
 \\  \nonumber & \quad  +N  \int_{\R^{3N}} \left(   k^2 \calG(k_0, k, k_2, \vec k)^2 - \left( \frac {2m}{m+1}\right)^2 \frac {1}{|k-p_\downarrow|^2} \right)  |\hat \xi(k_0+ k ,  k_{2}, \vec k)|^2  \dI k_0 \dI k_2 \dI \vec k 
 \\ \nonumber & \quad  + 2 k^2 N \Re \int_{\R^{3N}}  \hat \phi^*(k_0 , k, k_2,  \vec k)\calG(k_0, k, k_2, \vec k)  \hat \xi(k_0 + k ,  k_{2}, \vec k)    \dI k_0 \dI k_2 \dI \vec k
  \\  & \quad  + 2 k^2 N (N-1) \Re \int_{\R^{3N}}  \hat \phi^*(k_0 , k, k_2, \vec k)\calG(k_0, k, k_2, \vec k)  \hat \xi(k_0 + k_2  , k, \vec k)    \dI k_0 \dI k_2 \dI \vec k
  \label{ndown}
\end{align}
The terms $M_1^\downarrow$, $M_2^\downarrow$, $M_3^\downarrow$, $M_5^\downarrow$  and $M_6^\downarrow$ can be treated in the same way as the analogous terms in \eqref{nup} above. Eq.~\eqref{symfor} holds with $M_5^\downarrow$ in place of $M_3^\uparrow$ with $R_\uparrow$ replaced by 
\begin{equation}
R_\downarrow(k, k_1, k_2, \vec k) =  k^2   \calG(k_1-k, k, k_2 , \vec k)^2  - \left( \frac {2m}{m+1}\right)^2 \frac {1}{|k-p_\downarrow|^2} -  \left( \frac {2m}{m+1}\right)^2 \frac {2}{|k|^4} k \cdot \left( \frac {2}{m+1} k_1 - p_\downarrow \right) 
\end{equation}
which also satisfies the bound \eqref{readc}. 
The expression $M_4^\downarrow$ equals
\begin{equation}
M_4^\downarrow(k) =  k^2 N(N-1) \int_{\R^{3N}}  \calG(k_0-k_2 , k, k_2,  \vec k)^2  |\hat \xi(k_0,  k , \vec k)|^2 \dI k_0 \dI k_2 \dI \vec k 
\end{equation}
Performing the integration over $k_2$, one readily checks that
\begin{equation}
M_4^\downarrow(k) \leq \const  |k| N(N-1) \int_{\R^{3N}}    |\hat \xi(k_0,  k , \vec k)|^2 \dI k_0 \dI \vec k
\end{equation}
which is in $L^1(\R^3)$ since $\xi \in H^{1/2}(\R^{3N})$. Finally, using Cauchy-Schwarz in $(k,k_2,\vec k)$, 
\begin{equation}
\int_{\R^3} | M_7^\downarrow(k) | \dI k \leq 4N(N-1) \|\xi\|_{L^2(\R^{3N})}  \int_{\R^3} \| \hat \phi(k_0, \,\cdot\,)\|_{L^2(\R^{3N})} \dI k_0
\end{equation}
which is finite for $\phi\in H^2(\R^{3(N-1)})$, as remarked above. We  conclude, therefore, that also  $ k^2 n_\downarrow(k) - {\con}|k-p_\downarrow|^{-2}$ is integrable.

Since all the terms in \eqref{nup} and \eqref{ndown} are integrable, we can do the integration over $k$ term by term. For all the terms except $M_3^\uparrow$ and $M_5^\downarrow$, we have actually shown that the $L^1$-property holds even if the respective integrands are replaced by their absolute value, and hence we can freely use Fubini's theorem for these terms.  In the  form \eqref{symfor} (and the analogous expression for $M_5^\downarrow$) the same applies to $M_3^\uparrow$ and $M_5^\downarrow$, in fact.

For the norm of $\psi$, we shall write
\begin{align}\nonumber
\|\psi\|_{L^2(\R^{3(N+1)})}^2  & = \sum_{j=1}^4 n_ j  \\ \nonumber & = \|\phi\|_{L^2(\R^{3(N+1)})}^2 + 2 \Re \langle \phi | \calG \xi \rangle 
\\ \nonumber & \quad -  N(N-1) \int_{\R^{3N}}  \calG(k_0, k_1, k_2, \vec k)^2  \hat \xi^*(k_0 + k_1 ,  k_{2}, \vec k)\hat\xi(k_0 + k_2 , k_1, \vec k) \dI k_0 \dI k_1 \dI k_2 \dI \vec k  
\\  & \quad +  N \int_{\R^{3N}}  \calG(k_0, k_1, k_2, \vec k)^2  | \hat \xi(k_0 + k_1 ,  k_{2}, \vec k)|^2 \dI k_0 \dI k_1 \dI k_2 \dI \vec k 
\end{align}
We have 
\begin{equation}\label{n1}
 \int_{\R^3}  \left(  \frac 1 {2m}M_1^\uparrow(k)  +  \frac 12   M_1^\downarrow(k) \right) \dI k + \mu n_1  =  \left\langle \phi \left| - \frac 1 {2 m} \Delta_{x_0}    - \frac 1 2 \sum_{i=1}^N \Delta_{x_i} +\mu   \right| \phi\right\rangle  
\end{equation}
and
\begin{equation}\label{n3}
 \int_{\R^3}  \left[ \frac 1 {2m}  M_2^\uparrow(k) + \frac 12 \left( M_2^\downarrow(k) + M_3^\downarrow(k) \right) \right]  \dI k  + \mu n_3 =  - N \calT_{\rm off} (\xi)
\end{equation}
Moreover, we claim that 
\begin{equation}\label{n4}
 \int_{\R^3}  \left[ \frac 1 {2m}  M_3^{\uparrow}(k) + \frac 12 \left( M_4^\downarrow(k) + M_5^{\downarrow}(k)  \right)\right]  \dI k + \mu n_4 = - N  \calT_{\rm diag} (\xi)
 \end{equation}
To see this, note that we can replace $M_3^{\uparrow}(k)$ by its symmetrized version $\frac 12( M_3^{\uparrow}(k) +M_3^{\uparrow}(-k))$, and likewise for $M_5^\downarrow$. Then \eqref{n4} follows from the fact that  
\begin{align}\nonumber
 &\int_{\R^3} \left( \frac 1{4m} \left( R_\uparrow (k,k_1,\dots,k_N) + R_\uparrow (-k,k_1,\dots,k_N) \right) + \frac 1{4} \left( R_\downarrow (k,k_1,\dots,k_N) + R_\downarrow (-k,k_1,\dots,k_N)\right)  \right. \\ & \left.  \qquad + \frac 12 \sum\nolimits_{j=2}^N k_j^2  \calG(k_1-k, k, k_2,  \ldots, k_N)^2 \right) \dI k = - \calL(k_1,\ldots,k_N) 
\end{align}
which, in turn, uses that 
\begin{equation}
\int_{\R^3}  \left( \frac 2{ |k|^{2}}  - \frac 1 { |k-p |^{2}}   -  \frac 1 {|k+p|^{2}} \right) \dI k = 0
\end{equation}
for any $p \in \R^3$ (which can be proved, e.g., by computing the Fourier transform). 
Finally,
 \begin{align}\nonumber 
 & \int_{\R^3}  \left[ \frac 1 {2m}  M_4^\uparrow(k) + \frac 12 \left( M_6^\downarrow(k) + M_7^\downarrow(k)  \right)\right]  \dI k + \mu n_2 \\ & = 2 N  \Re  \int \hat \phi^*(k_0,k_1-k_0, k_2, \vec k) \hat \xi(k_1,k_2, \vec k) \dI k_0 \dI k_1 \dI k_2 \dI \vec k  \label{w7}
 \end{align}
In Fourier space, the boundary condition \eqref{bc} satisfied by $\phi$ reads  
\begin{equation}
\int \hat \phi(k_0,k_1-k_0, k_2, \vec k)  \dI k_0 =    \left(  \frac {2m}{m+1}\alpha\hat \xi + \widehat{\calGam\xi} \right)   (k_1,k_2, \vec k)
\end{equation}
and hence
\begin{equation}\label{n2}
\eqref{w7} = 2 N \left( \calT_{\rm diag} (\xi) + \calT_{\rm off} (\xi) + \frac {2m}{m+1}  \alpha \|\xi\|_{L^2(\R^{3N})}^2  \right)
\end{equation}
A combination of \eqref{n1}, \eqref{n3}, \eqref{n4}, \eqref{n2} with \eqref{def:cFal}  establishes \eqref{eq:tan} and thus completes the proof of Theorem~\ref{thm:tan}. \hfill\qed

\begin{remark}
The proof of Theorem~\ref{thm:tan} does not actually make use of the assumption $\xi\in H^{1}(\R^{3N})$, it is only used that 
\begin{equation}\label{condxi}
\int_{\R^{3N}} \left( 1 + \sum_{j=1}^N |k_j|^2 \right)^{1/2} \ln \left( 2 + \sum_{j=1}^N |k_j|^2\right) | \hat \xi(k_1,\ldots,k_N)|^2 \dI k_1 \cdots \dI k_N < \infty
\end{equation}
By Theorem~\ref{thm:domain}, this is actually the case if $\Lambda_0(m) = 2\Lambda(m) <1$ (instead of $\Lambda_1(m)<1$) since then, by continuity, $\Lambda_\beta(m)<1$ for some $\beta>0$, and hence $\xi \in H^{(1+\beta)/2}(\R^{3N})$.
\end{remark}

\section*{Acknowledgments}
Financial support by the 
 European Research Council (ERC) under the European Union's Horizon 2020 research and innovation programme (grant agreement No 694227), and by the Austrian Science Fund (FWF), project Nr. P 27533-N27, is gratefully acknowledged.

\end{document}